\documentclass[aps,showpacs,pra,twocolumn]{revtex4-1}

\usepackage{exscale}
\usepackage{bbm}
\usepackage{graphicx}
\usepackage{amsmath}
\usepackage{latexsym}
\usepackage{amsfonts}
\usepackage{amssymb}
\usepackage{times}
\usepackage[T1]{fontenc}
\usepackage{amsthm}
\usepackage{enumerate}
\usepackage{bbold}
\usepackage{color}
\usepackage[colorlinks=true,citecolor=blue,urlcolor=blue]{hyperref}
\usepackage{mathtools}
\usepackage{changes}
\usepackage[normalem]{ulem}
\usepackage{amsfonts,amsmath,amssymb,amsthm}

\newcommand{\ket}[1]{|#1\rangle}
\newcommand{\bra}[1]{\langle#1|}

\theoremstyle{plain}
\newtheorem{thm}{Theorem}
\newtheorem{lem}{Lemma}
\newtheorem{fakt}{Fact}
\newtheorem*{defin}{Definition}

\begin{document}

\title{Self-testing maximally-dimensional genuinely entangled subspaces within the stabilizer formalism}
\author{Owidiusz Makuta}
\affiliation{Center for Theoretical Physics, Polish Academy of Sciences, Aleja Lotnik\'{o}w 32/46, 02-668 Warsaw, Poland}

\author{Remigiusz Augusiak}

\affiliation{Center for Theoretical Physics, Polish Academy of Sciences, Aleja Lotnik\'{o}w 32/46, 02-668 Warsaw, Poland}

\begin{abstract}
Self-testing was originally introduced as a device-independent method of certification of entangled quantum states and local measurements performed on them. Recently, in [F. Baccari \textit{et al.}, \href{https://arxiv.org/abs/2003.02285}{arXiv:2003.02285}] the notion of state self-testing has been generalized to entangled subspaces and the first self-testing strategies for exemplary genuinely entangled subspaces have been given. The main aim of our work is to pursue this line of research and to address the question how "large" (in terms of dimension) are genuinely entangled subspaces that can be self-tested, concentrating on the multiqubit stabilizer formalism. To this end, we first introduce a framework allowing to efficiently check whether a given stabilizer subspace is genuinely entangled. Building on it, we then determine the maximal dimension of genuinely entangled subspaces that can be constructed within the stabilizer subspaces and provide an exemplary construction of such maximally-dimensional subspaces for any number of qubits. Third, we construct Bell inequalities that are maximally violated by any entangled state from those subspaces and thus also any mixed states supported on them, and we show these inequalities to be useful for self-testing. Interestingly, our Bell inequalities allow for identification of higher-dimensional face structures in the boundaries of the sets of quantum correlations in the simplest multipartite Bell scenarios in which every observer performs two dichotomic measurements.
\end{abstract}

\maketitle

\section{Introduction}
Nonlocality, or the existence of correlations that violate Bell inequalities, is one of the most characteristic features of quantum theory \cite{review}. Initially studied as an indicator of impossibility of describing quantum correlations in terms of local-realistic models, later nonlocality has been turned into a powerful resource for certain applications such as for instance device-independent certification. Indeed, nonlocality generated by an unknown quantum system allows one to make nontrivial statements about this system. For instance, observation of nonlocal correlations in an unknown quantum system can be used to certify in a device-independent way its dimension \cite{dimension}, that it is entangled \cite{DIentanglement}, and that the outcomes of quantum measurements performed on it are truly random \cite{randomness1,randomness2}. 

Interestingly, it turns out that in certain situations the observed nonlocal correlations enable, basically complete, characterization of the quantum system under study along with quantum measurements performed on it. This strongest form of device-independent certification is referred to as self-testing \cite{MY1,MY2} (see also the recent review \cite{self-review}), and due to its utility certification of quantum systems it has attracted much attention in recent years. Device-independent certification strategies have been designed for various quantum states \cite{McKague,YangNavascues13,Coladangelo,Shubhayan} and/or quantum measurements \cite{obs1,Quantum}. Extensions of self-testing to other scenarios featuring nonclassical effects have also been studied, such as the steering scenario \cite{steering}, the prepare-and-measure scenario \cite{PM}, or the contextuality scenario \cite{context1,Knill,context2}.

Very recently, in Ref. \cite{subspaces} a definition of self-testing of entangled subspaces has been put forward as a natural generalization of state self-testing, and self-testing schemes for genuinely entangled subspaces obtained within the stabilizer formalism were provided such as the one corresponding to the well-known toric code \cite{toric}. Our main aim here is to pursue this research direction and address one of questions Ref. \cite{subspaces} leaves open, namely, what is the maximal dimension of a genuinely entangled subspace that can be self-tested
from maximal violation of Bell inequalities. For simplicity here we focus on subspaces that can be obtained within the stabilizer formalism. To reach our goal, we first introduce a simple and efficient way of checking whether a given stabilizer subspace is genuinely multipartite entangled. Building on it, we then establish what is the maximal dimension of a genuinely entangled subspace within the stabilizer formalism proving the conjecture posed recently in Ref. \cite{Waegell}, and present an exemplary construction of such maximally-dimensional subspaces for any number of qubits. Third, we show that each such subspace can be self-tested from maximal violation of a multipartite Bell inequality constructed with the aid of the approach presented in Refs. \cite{subspaces,graphstates}. We finally use these Bell inequalities to identify nontrivial faces in the boundaries of
the sets of quantum correlations in the simplest Bell scenarios involving two dichotomic observables per observer.

It should be noticed that this work is basically a publication version of Ref. \cite{Magisterka} with a few modifications and additions such as 
the construction of GME subspaces presented in Sec. \ref{SecIIIB}
along with the associated Bell inequalities presented in Appendix \ref{AppB}. 
Also, the proof contained in Sec. \ref{boundaries} is a significantly improved version of that presented in Ref. \cite{Magisterka}.

\section{Preliminaries}
\label{prelim}

Before presenting our results we first introduce
some notions and terminology used throughout our work.
We begin with the notion of genuine multipartite entanglement, 
and then move to Bell scenarios and the stabilizer formalism. 
A reader familiar with these notions can skip this section.

\textit{(1) Genuinely entangled subspaces.} Let us consider $N$ spatially separated parties $P_i$ and assume that they share some pure quantum state $\ket{\psi}$ belonging to a product Hilbert space 
\begin{equation}
    \mathcal{H}_P=\mathcal{H}_{P_1}\otimes\ldots\otimes\mathcal{H}_{P_N},
\end{equation}
where the local Hilbert space $\mathcal{H}_{P_i}$ corresponds to party $P_i$ and we assume it to be finite-dimensional.

Consider then a partition of $N$ observers into two disjoint groups $Q$ and $\overline{Q}$ such that $Q\cup\overline{Q}=\{1,\ldots,N\}$ and denote it $Q|\overline{Q}$. We call such a bipartition nontrivial if both sets $Q$ and $\overline{Q}$ are nonempty. We say that $\ket{\psi}$ is \textit{entangled across} $Q|\overline{Q}$ if it cannot be written as $\ket{\psi}=\ket{\psi_Q}\otimes \ket{\psi_{\overline{Q}}}$ for some pure states $\ket{\psi_Q}$ and $\ket{\psi_{\overline{Q}}}$ supported on $Q$ and $\overline{Q}$, respectively. We then call $\ket{\psi}$ \textit{genuinely multipartite entangled} (GME) if it is entangled across any nontrivial bipartition $Q|\overline{Q}$.

Consider then a notrivial subspace $V\subset \mathcal{H}_P$. Following Ref. \cite{Demianowicz} we call $V$ \textit{genuinely multipartite entangled} (or, shortly, \textit{genuinely entangled}) if all pure states belonging to it are genuinely multipartite entangled. The notion of genuinely entangled subspaces is a natural multipartite generalization of the notion of completely entangled subspaces introduced in the bipartite case \cite{Bhat,Parthasarathy}. A well-known example of such a genuinely entangled subspace in $(\mathbbm{C}^2)^{\otimes 3}$ is the one spanned by the three-qubit Greenberger-Horne-Zeilinger (GHZ) state and the $W$ states
\begin{equation}
\ket{\mathrm{GHZ}_3}=\frac{1}{\sqrt{2}}(\ket{0}^{\otimes 3}+\ket{1}^{\otimes 3}),
\end{equation}
\begin{equation}
\ket{W_3}=\frac{1}{\sqrt{3}}(\ket{001}+\ket{010}+\ket{100}).
\end{equation}
It is not difficult to verify that for any $\alpha,\beta\in\mathbbm{C}$ such that $|\alpha|^2+|\beta|^2=1$, the state $\alpha\ket{\mathrm{GHZ}_3}+\beta\ket{W_3}$ is genuinely entangled.

\textit{(2) The Bell scenario.} Let us now assume that the parties $P_i$ share a mixed state $\rho$ acting on $\mathcal{H}_P$ and that each party on their subsystem of $\rho$ can freely choose to measure one of two quantum observables $A_0^{(i)}$ or $A_1^{(i)}$. Here we consider the simplest case of all observables being dichotomic and, traditionally, denote their eigenvalues by $\pm 1$, meaning in particular that $[A_{x_{i}}^{(i)}]^2=\mathbbm{1}_{P_i}$ for any $x_i=0,1$ and $i=1,\ldots,N$, where $\mathbbm{1}_{P_i}$ is the identity acting on $\mathcal{H}_{P_i}$. 

If this procedure is repeated sufficiently many times, the parties can estimate the collection of expectation values 
\begin{equation}\label{correlations}
    \langle A_{x_{i_1}}^{(i_1)}\cdots  A_{x_{i_k}}^{(i_k)}\rangle = \textrm{Tr}\left[\left(A_{x_{i_1}}^{(i_1)}\otimes \cdots \otimes A_{x_{i_k}}^{(i_k)}\right)\rho \right]
\end{equation}
with $i_1<i_2<\ldots,<i_k$, $k=1,\ldots,N$ and $i_j=1,\ldots,N$. In what follows
we denote this collection by $\mathcal{P}$ and refer to it as to correlations
or behaviour. It is worth adding here that in the considered scenario all correlations obtained by performing measurements on quantum states
form a convex set which we denote $\mathcal{Q}_{N}$. Later in Sec. \ref{boundaries} we identify faces in the boundaries of $\mathcal{Q}_{N}$.

We say that the correlations $\mathcal{P}$ are \textit{nonlocal} 
if they cannot be reproduced by any local hidden variable model, or, equivalently, if they violate Bell inequalities whose most general form in the considered scenario is
\begin{equation}
I_N := \sum_{k=1}^N 
\sum_{\substack{1 \leq i_1<i_2<\ldots<i_k \leq N  \\  x_{i_1},\ldots,x_{i_N}=0,1 }}\alpha^{i_1,\ldots,i_k}_{x_{i_1},\ldots,x_{i_k}}\langle A_{x_{i_1}}^{(i_1)} \ldots A_{x_{i_k}}^{(i_k)}\rangle \leq \beta_c,
\end{equation}
where all $\alpha^{i_1,\ldots,i_k}_{x_{i_1},\ldots,x_{i_k}}\in\mathbbm{R}$
and $\beta_c$ is the maximal value of the Bell expression $I_N$ over all local deterministic correlations, that is, those for which all expectation values 
are product 
\begin{equation}\label{locdet}
\langle A_{x_{i_1}}^{(i_1)} \ldots A_{x_{i_k}}^{(i_k)}\rangle =
\langle A_{x_{i_1}}^{(i_1)}\rangle\cdot \ldots\cdot \langle A_{x_{i_k}}^{(i_k)}\rangle
\end{equation}
and each $\langle A_{x_{i_k}}^{(i_k)}\rangle=\pm1$.

A simple example of a Bell inequality is the one due to 
Clauser, Horne, Shimony and Holte, typically referred to as 
CHSH Bell inequality \cite{CHSH}. 
It is defined in the bipartite scenario ($N=2$) and reads
\begin{equation}\label{CHSH}
    \langle A_0^{(1)} A_0^{(2)}\rangle+\langle A_0^{(1)} A_1^{(2)}\rangle+\langle A_1^{(1)} A_0^{(2)}\rangle-\langle A_1^{(1)} A_1^{(2)}\rangle\leq 2.
\end{equation}
The maximal quantum value of the CHSH Bell inequality is 
$2\sqrt{2}$ and is achieved by 
the two-qubit maximally entangled state $\ket{\psi_{2}^+}=(1/\sqrt{2})(\ket{00}+\ket{11})$ and
the following observables 
\begin{equation}
A_0^{(1)}=\frac{1}{\sqrt{2}}(X+Z),\qquad 
A_1^{(1)}=\frac{1}{\sqrt{2}}(X-Z),
\end{equation}
and $A_0^{(2)}=X$ and $A_1^{(2)}=Z$, where $X$ and $Z$ are the well known Pauli matrices.  It is worth noting that 
after plugging these into the CHSH Bell expression
the resulting Bell operators is a sum of the stabilizing operators
$X\otimes X$ and $Z\otimes Z$ of $\ket{\psi_{2}^+}$. Later we exploit this observation when constructing our Bell inequalities maximally violated by entangled subspaces. 

Let us finally stress that nonlocality is the key resource making the device-independent verification and certification possible: the observers are able to make nontrivial statements about functioning of their devices only if
they observe quantum nonlocality \cite{review}.

\textit{(3) Stabilizer formalism.} Let $X$, $Y$, $Z$ be the Pauli matrices and $G_N$ denote the Pauli group on $N$ qubits which consists of $N$-fold tensor products of $X$, $Y$, $Z$ or $\mathbbm{1}_2$ with the overall factors $\pm 1$ or $\pm i$, where by $\mathbbm{1}_2$ we denote the $2\times 2$ identity matrix.

Let us consider a subgroup $\mathbb{S}$ of $\mathbb{G}_N$ generated by $k$ independent elements of $\mathbb{G}_N$, denoted $G_i$ with $i=1,\ldots,k$. Such a subgroup is called \textit{stabilizer} if it stabilizes a nontrivial subspace $V$ in $(\mathbbm{C}^2)^{\otimes N}$, that is, 
\begin{equation}
G_i\ket{\psi}=\ket{\psi}\qquad (i=1,\ldots, k)
\end{equation}
for any $\ket{\psi}\in V$. In what follows we refer to a subspace stabilized in the above sense as \textit{stabilizer subspace}.  

Importantly, not every subgroup of $\mathbb{G}_N$ stabilizes a nontrivial subspace; for instance, $X\otimes X$ and $-X\otimes X$ generate a subgroup that stabilizes no subspace because it contains $-\mathbbm{1}_2^{\otimes 2}$. It turns out that a necessary and sufficient condition for a subgroup $\mathbb{S}$ to stabilize a nontrivial subspace is that it is Abelian
(or, equivalently, that all generators $G_i$ mutually commute), and that $-\mathbbm{1}_2^{\otimes N}\notin \mathbb{S}$ (see Ref. \cite{NielsenChuang}). Let us notice here that, as proven below, it actually suffices to satisfy the second condition in order for $\mathbb{S}$ to stabilize a nontrivial subspace.
\begin{fakt}
Consider a subgroup $\mathbb{S}=\langle G_1,\ldots,G_k\rangle\subseteq \mathbb{G}_N$.
If $-\mathbbm{1}_2^{\otimes N}\notin \mathbb{S}$, then the generators $G_i$ mutually commute.
\end{fakt}
\begin{proof}
Assume that there is a pair of generators $G_i,G_j$ which do not commute, $[G_i,G_j]\neq 0$.
Due to the definition of $\mathbb{G}_N$ it is not difficult to see that they must anticommute, 
\begin{equation}\label{Raimat}
    \{G_i,G_j\}=0.
\end{equation}
Consider then a product $G_iG_jG_iG_j$. Taking into account 
Eq. (\ref{Raimat}) it can be rewritten as $G_iG_jG_iG_j=-G_i^2G_j^2$.
Now, we know that $G_i^2=G_j^2=\mathbbm{1}_2^{\otimes N}$ (the 
cases $G_i^2=-\mathbbm{1}_2^{\otimes N}$ and $G_i^2=-\mathbbm{1}_2^{\otimes N}$ are ruled out by the assumption), and therefore $G_iG_jG_iG_j=-G_i^2G_j^2=-\mathbbm{1}_2^{\otimes N}$ which contradicts the assumption.
\end{proof}

The dimension of a subspace stabilized by a subgroup satisfying these conditions is a simple function of the number of its generators $k$ as well as the number of qubits $N$: $D_{N,k}=2^{N-k}$; the number of elements of the stabilizer is $|\mathbb{S}|=2^{k}$. Let us also notice there that a stabilizer subspace of dimension $2^k$ might be used to encode $k$ logical qubits; the corresponding vectors belonging to $V_N$ are called quantum code words (for more details see Refs. \cite{GottesmanPhD,Steane,Shor}). 

Let us illustrate these concepts with two examples of states and subspaces stabilized by subgroups of $\mathbb{G}_N$. First we consider the well-known graph states. Consider a graph $G=(V,E)$
where $V$ is the set of vertices matching the parties $P_i$ and $E$ stands for the set of edges connecting vertices.
To each vertex $i$ one associates the following operator 
\begin{equation}\label{StabGrSt}
G_i=X_i\otimes \bigotimes_{j\in N(i)}Z_j,
\end{equation}
where $X_i$ denotes $X$ operator acting on qubit $i$, whereas
$N(i)$ stands for the set of neighbours of the vertex $i$, which are all vertices of $V$ connected to $i$ by an edge. These operators generate a subgroup of $\mathbb{G}_N$ that for any $N$ stabilizes a one-dimensional subspace spanned by the so-called graph state $\ket{\psi_G}$. 
It is worth noting here that Bell inequalities maximally violated by the graph states are known \cite{graphstates0}. Moreover, scalable self-testing schemes for them have recently been given in Ref. \cite{graphstates}.

Let us now provide an example of a stabilizer generating a GME subspace of dimension higher than one. To this aim, we consider
the stabilizer corresponding to the five-qubit error correction code (see, e.g., \cite{NielsenChuang}), which is generated by the following four operators acting on  $(\mathbbm{C}^{2})^{\otimes 5}$:
\begin{eqnarray}\label{Stab5Qubit}
    G_1&=&X_1Z_2Z_3X_4,\nonumber\\
    G_2&=&X_2Z_3Z_4X_5,\nonumber\\
    G_3&=&X_1X_3Z_4Z_5,\nonumber\\
    G_4&=&Z_1X_2X_4Z_5,
\end{eqnarray}
where for simplicity we skip the tensor product.
By direct inspection one confirms that all these operators mutually commute and that one cannot generate $-\mathbbm{1}^{\otimes 5}_2$ by taking their product. Thus, $\mathbb{S}_{\mathrm{5-qubit}}=\langle G_1,G_2,G_3,G_4\rangle$ spans a nontrivial subspace whose dimension is $2^{5-4}=2$, as the number of independent generators is four. This subspace is genuinely entangled and a self-testing method based on maximal violation of a Bell inequality has recently been designed in Ref. \cite{subspaces}. One of our aims here is to generalize the results of the latter work and construct stabilizer subspaces of maximal dimension for any number of qubits which can be self-tested from maximal violation of certain multipartite Bell inequalities.

\section{Criteria for stabilizer subspaces to be genuinely entangled}
\label{SecIII}

In our work we aim at providing self-testing statements for stabilizer subspaces that are genuinely entangled, whereas it is know that not every stabilizer subspace is genuinely entangled. An example of a non-GME stabilizer subspace is the one spanned by the codewords of the famous Shor's code \cite{Shor}, that is, $\ket{\mathrm{GHZ}_3}^{\otimes 3}$ and $\ket{\mathrm{GHZ}_3^{-}}^{\otimes 3}$ with  
\begin{equation}
    \ket{\mathrm{GHZ}_3^{-}}=\frac{1}{\sqrt{2}}(\ket{000}-\ket{111}).
\end{equation}

Therefore, we need a condition allowing to check whether a given 
stabilizer subspace is genuinely multipartite entangled. A simple sufficient for a subspace to be GME condition that does the job was recently formulated in Ref. \cite{subspaces}. Here we show it to be necessary too. Building on it we then introduce another necessary and sufficient condition for a stabilizer subspace to be GME, which not only allows for deciding whether a stabilizer subspace is GME or not, but can also be used to construct such subspaces. 

Consider a stabilizer $\mathbb{S}$ generated 
by $k$ stabilizing operators $G_i$ $(i=1,\dots,k)$ and 
a nontrivial bipartition $Q|\overline{Q}$.
With respect to this bipartition each generator can be written as
\begin{equation}\label{division}
    G_i=G_i^Q\otimes G_i^{\overline{Q}},
\end{equation}
where each $G_i^{Q}$ $(G_i^{\overline{Q}})$ acts on the Hilbert space associated to the group $Q$ ($\overline{Q}$). 

Let us now formulate our criterion.

\begin{thm}\label{fakt3}
Consider an abelian subgroup $\mathbb{S}$ of $\mathbb{G}_N$ generated by $G_i$ $(i=1,\ldots,k)$ such that $-\mathbbm{1}^{\otimes N}_2\notin\mathbb{S}$. 
For every nontrivial bipartition $Q|\overline{Q}$ 
there exist $i\neq j$ such that 
\begin{equation}\label{thm1cond}
\left\{G_i^{Q},G_j^{Q}\right\}=0
\end{equation}
iff the subspace $V$ stabilized by $G_i$ is genuinely multipartite entangled.
\end{thm}
Before we prove this theorem it is worth noting that
the condition (\ref{thm1cond}) can be reformulated as 
$\{G_i^{\overline{Q}},G_j^{\overline{Q}}\}=0$, which is due 
to the fact that all $G_i$ mutually commute, so for any bipartition
$Q|\overline{Q}$ either 
\begin{equation}\label{Com}
\left[G_i^Q,G_j^Q\right]=0\quad\mathrm{and}\quad   \left[G_i^{\overline{Q}},G_j^{\overline{Q}}\right]=0,
\end{equation}
or
\begin{equation}\label{antiCom}
 \left\{G_i^Q,G_j^Q\right\}=0\quad\mathrm{and}\quad   \left\{G_i^{\overline{Q}},G_j^{\overline{Q}}\right\}=0
\end{equation}
holds true.
\begin{proof}The proof of the "if" part was given in Ref. \cite{subspaces}, however, to make the paper self-contained, we also present it here. Let us assume that the subspace $V$ stabilized by the subgroup $\mathbb{S}$ is not genuinely entangled. This means that $V$ contains a non-genuinely entangled state, denoted $\ket{\psi}$, which for some nontrivial bipartition $Q|\overline{Q}$ can be written as 
\begin{equation}\label{vir2}
    \ket{\psi}=\ket{\psi_Q}\otimes\ket{\psi_{\overline{Q}}},
\end{equation}
for some pure states $\ket{\psi_Q}$ and $\ket{\psi_{\overline{Q}}}$. From the fact that $\ket{\psi}$ is stabilized by every generator $G_i$,
one concludes by virtue of Eqs. (\ref{division}) and (\ref{vir2}) that with respect to the bipartition $Q|\overline{Q}$,
\begin{equation}
    G_i^Q\ket{\psi_Q}=\mathrm{e}^{\mathrm{i}\phi_i}\ket{\psi_Q},
\end{equation}
for any $i$, where $\phi_i\in\mathbbm{R}$.
This in turn implies that for any pair $G_i^Q$
and $G_j^Q$ with $i\neq j$, 
\begin{equation}
    \left\{G_i^Q,G_j^Q\right\}\ket{\psi}=2\mathrm{e}^{\mathrm{i}
    (\phi_i+\phi_j)}\ket{\psi_Q},
\end{equation}
which contadicts the assumption that for the considered bipartition there exists a pair of anticommuting operators $G_i^Q$, $G_j^Q$.

Let us now move on to the "only if" part. Assume that $V$ is genuinely multipartite entangled and that there exists a nontrivial bipartition $Q|\overline{Q}$ such that for every pair of generators $G_i,G_j$ with $i\neq j$, 
\begin{equation}
    \left\{G_i^Q,G_j^Q\right\}\neq 0.
\end{equation}
Due to the fact that all generators commute, this implies that
\begin{equation}\label{com1}
    \left[G_i^Q,G_j^Q\right]=0\quad\mathrm{and}\quad \left[G_i^{\overline{Q}},G_j^{\overline{Q}}\right]=0
\end{equation}
for every pair $i\neq j$ [cf. Eqs. (\ref{Com}) and (\ref{antiCom})]. 
As a result all $G_i^Q$ (as well as $G_i^{\overline{Q}}$) are diagonal in the same basis, that is, there exist orthonormal bases
$\{\ket{\phi_i}\}$ and
$\{\ket{\varphi_i}\}$ spanning the Hilbert spaces
corresponding to the groups $Q$ and $\overline{Q}$ such that
\begin{equation}
    G_i^Q\ket{\phi_j}=\lambda_j^{(i)}\ket{\phi_j}.
\end{equation}
and
\begin{equation}
    G_i^{\overline{Q}}\ket{\varphi_j}=\widetilde{\lambda}_j^{(i)}\ket{\varphi_j}
\end{equation}
for some $\lambda_j^{(i)},\widetilde{\lambda}_j^{(i)}\in\mathbbm{R}$.
Clearly, the product vectors $\ket{\phi_i}\ket{\varphi_j}$ form a basis of the total Hilbert space $(\mathbbm{C}^2)^{\otimes N}$ and at the same time are eigenvectors of all the stabilizing operators $G_i$. Consequently, the subspace $V$ contains pure states that are separable across the bipartition
$Q|\overline{Q}$, unless the space stabilized by $G_i$ is empty. In both cases we are in contradiction with the assumptions, which completes the proof.
\end{proof}

Let us emphasize that Theorem \ref{fakt3} gives us a way of determining whether a subspace $V$ stabilized by $\mathbb{S}$ is GME by looking at its generators instead of verifying whether each vector from $V$ is GME. While this already significantly facilitates checking if a given stabilizer subspace is genuinely entangled, in what follows we reformulate Theorem \ref{fakt3} to further simplify this task.

\subsection{Another iff criterion for a stabilizer subspace to be GME}

Let us consider the $\mathbb{Z}_2$ field, which consists of the set $\{0,1\}$ equipped with addition and multiplication mod $2$. Consider then a vector space $F_N=\mathbb{Z}_2^N$ over the field $\mathbb{Z}_2$, consisting of $N$-dimensional vectors whose entries are from $\{0,1\}$. By $e_i$ we denote the standard basis of $F_i$ consisting of vectors with exactly one $1$ at the $i$th position and zeros elsewhere. Notice that for any $f\in F_N$, $f\oplus f$=0.

Now, we consider a subset of $F_N$, denoted $W_N$ and composed of those elements of $F_N$ that have even number of ones. It is not difficult to see that 
$W_N$ is a proper subspace of $F_N$ because addition of two ''even''
vectors from $W_N$ can result only in an "even" vector. The dimension of $W_N$ is $\dim W_N=N-1$ because every vector containing an even number of ones can be written as a linear combination of vectors 
\begin{eqnarray}\label{vectorsW}
    w_1&=&e_1+e_2\equiv (1,1,0,\ldots,0),\nonumber\\
    w_2&=&e_2+e_3\equiv (0,1,1,0,\ldots,0),\nonumber\\
    &\vdots &\nonumber\\
    w_{N-1}&=&e_{N-1}+e_N\equiv (0,\ldots,0,1,1)
\end{eqnarray}
containing two ones, there is $N-1$ such vectors and they are linearly independent. 

Let us now make use of this formalism to introduce our criterion. 
To this end, consider again a stabilizer $\mathbb{S}$ generated by $G_i$
$(i=1,\ldots,k)$. To each pair $G_i$ and $G_j$ with $i\neq j$, we associate
a vector from $F_N$ defined as
\begin{equation}\label{vectors}
v_{i,j}=\sum_{k=1}^N e_k \delta_k(G_i,G_j),
\end{equation}
where $\delta_k(G_i,G_j)$ is a two-valued real function such that
$\delta_k(G_i,G_j)=0$ or $\delta_k(G_i,G_j)=1$ if the local Pauli operators appearing at the $k$th position in $G_i$ and $G_j$ commute or anticommute, respectively. In other words, $v_{ij}$ is a vector that contains all information about commutation/anticommutation of local Pauli matrices in $G_i$ and $G_j$. 

To illustrate how the vectors $v_{ij}$ are created with an example, let us consider a stabilizer generated by the following two operators
\begin{equation}
    G_1=X\otimes X\otimes \mathbbm{1},\qquad G_2=Z\otimes Z\otimes\mathbbm{1}.
\end{equation}
Clearly, the Pauli matrices appearing on the first two positions of $G_1$ and $G_2$ anticommute, whereas they commute at the third position and therefore the vector $v_{1,2}$ associated to $G_1$ and $G_2$ reads
\begin{equation}
    v_{1,2}=e_1+e_2\equiv (1,1,0).
\end{equation}

In general, given a stabilizer $\mathbb{S}$ generated by $k$ operators $G_k$ we associate to it $k(k-1)/2$ vectors $v_{ij} $which, due to the fact that all $G_i$ mutually commute, have even number of ones and thus all belong to $W_N$. In fact, for any stabilizer $\mathbb{S}$ these vectors span a subspace of $W_N$, which we denote $K(\mathbb{S})$.

Our aim now is to reformulate Theorem \ref{fakt3} in terms of subspaces $K(\mathbb{S})$. 
To do so, we need to introduce a few additional notions regarding bipartitions
of the set $\{1,\ldots,N\}$. Consider again a stabilizer $\mathbb{S}$ generated by $G_i$ and the associated vectors $v_{i,j}$. Consider then a nontrivial bipartition $Q|\overline{Q}$. We call it \textit{even} or \textit{odd} on a vector $v_{i,j}$ if this vector has, respectively, even or odd number of ones in $Q$ and $\overline{Q}$. 
We then call a bipartition $Q|\overline{Q}$ \textit{even on $K(\mathbb{S})$} if it is even on every vector from $K(\mathbb{S})$, and call it \textit{odd on} $K(\mathbb{S})$ if it is odd on at least one vector from $K(\mathbb{S})$.

Having this, we can now reformulate Theorem \ref{fakt3} in terms of the vectors from $K(\mathbb{S})$.
\begin{lem}\label{Lem0}
A stabiliser subspace $V$ is genuinely multipartite entangled iff all nontrivial bipartitions are odd on the corresponding $K(\mathbb{S})$.
\end{lem}
\begin{proof}Let us begin with the "if" part and assume that all nontrivial bipartitions on $K(\mathbb{S})$ are odd. This means, via our construction of $v_{i,j}$,
that the generators $G_{i}$ and $G_j$ anticommute on any bipartition, which 
via Theorem \ref{fakt3} means that the subspace $V$ is genuinely entangled.

For the "only if" part let us assume that the subspace $V$ is genuinely entangled. By virtue of Theorem \ref{fakt3} this implies that for any 
nontrivial bipartition $Q|\overline{Q}$ there exists a pair of stabilizing
operators $G_i$ and $G_j$ $(i\neq j)$ such that $G_i^Q$ and $G_j^Q$ anticommute, 
which implies that the corresponding vector $v_{i,j}$ has odd number of
ones in $Q$ and $\overline{Q}$, and thus any bipartition on $K(\mathbb{S})$ is odd.
\end{proof}

We are now ready to introduce our main necessary and sufficient criterion 
for a stabilizer subspace to be genuinely entangled.
To this end, we introduce the following bilinear form 
$h(\cdot,\cdot):F_N\times F_N\to \mathbb{Z}_2$ defined as
\begin{equation}
h(v,w)=\sum_{i=1}^N v_iw_i,
\end{equation}
where $v_i$ and $w_i$ are vectors elements of $v$ and $w$ in the standard basis $\{e_i\}$. 

Let us then notice that each bipartition $Q|\overline{Q}$ can be represented within this new formalism by a vector from $F_N$ defined as
\begin{equation}
\phi=\sum_{i\in Q}e_i.    
\end{equation}
That is, it has ones at those positions that belong to $Q$ and zeros elsewhere.
This representation gives us an easy way of determining whether a given bipartition is odd or even on a vector $v$: it is odd if $h(\phi,v)=1$ and even if $h(\phi,v)=0$. 

We can now prove the following lemma.
\begin{lem}\label{Lem1}
Let $u_1,\ldots,u_m$ be linearly independent vectors from $W_N$. If some bipartition $\phi$ is odd on $u_1+\ldots+u_m$, then it is odd on the set $\{u_1,\ldots,u_m\}$.
\end{lem}
\begin{proof}If $\phi$ is odd on the vector $u_1+\ldots+u_m$, then 
\begin{equation}
    h\left(\phi,u_1+\ldots+u_m\right)=\sum_{i=1}^{m}h (\phi,u_i)=1.
\end{equation}
The last equation implies that for at least one vector $u_i$, $h(\phi,u_i)=1$, which according to our definition means that the representation $\phi$ is odd on 
the set $\{u_1,\ldots,u_m\}$.
\end{proof}

Having all this we can formulate our main criterion.

\begin{thm}\label{Crit2}
A subspace $V\subset (\mathbbm{C}^2)^{\otimes N}$ stabilized by $\mathbb{S}$
is genuinely entangled iff $\dim K(\mathbb{S})=N-1$.
\end{thm}
\begin{proof}Let us begin with the "$\Rightarrow$" part of the proof and assume that the subspace $V$ is genuinely entangled but $\dim K(\mathbb{S})=m<N-1$. The fact that $V$ is genuinely entangled, implies by Lemma \ref{Lem0} that all the nontrivial bipartitions are odd on $K(\mathbb{S})$. 

Let us then consider some basis in $K(\mathbb{S})$, denoted $\{u_i\}_{i=1}^{m}$. Due to Lemma \ref{Lem1} every nontrivial bipartition is also odd on this basis. Since there are $2^{N-1}-1$ nontrivial bipartitions and $m<N-1$ basis vectors, there exist two different nontrivial bipartitions represented by $\phi_1,\phi_2\in F_d$ such that
\begin{equation}
    h(\phi_1,u_i)=h(\phi_2,u_i)
\end{equation}
for each $i=1,\ldots,m$. Consider now another bipartition $\phi=\phi_1+\phi_2$
which is also nontrivial because $\phi_i$ are nontrivial. Clearly,
\begin{equation}
h(\phi,u_i)=h(\phi_1,u_i)+h(\phi_2,u_i)=0    
\end{equation}
and consequently $\phi$ is even on the basis $\{u_1,\ldots,u_m\}$, which by virtue of Lemma \ref{Lem1} means that $\phi$ is also even on $K(\mathbb{S})$, contradicting that $V$ is genuinely entangled.

Let us now move on to the "$\Leftarrow$" part and assume that $\dim K(\mathbb{S})=N-1$. 
We will show that the latter implies that every nontrivial bipartition is odd on $K(\mathbb{S})$, which by virtue of Lemma \ref{Lem0} means that the subspace $V$ is genuinely entangled and the desired implication follows.

Let us therefore consider a nontrivial bipartition $Q|\overline{Q}$ represented by a vector $\phi\in F_N$ such that $\phi\neq (0,\ldots,0)$ and $\phi\neq (1,\ldots,1)$. The latter conditions imply that the vector $\phi$ contains at least one $0$ at a position $i_0$ and one $1$ at a position $i_1$. Without any loss of generality we can assume that $i_0\in \overline{Q}$ and $i_1\in Q$. Let us then consider the following vector
\begin{equation}
v=e_{i_0}+e_{i_1}, 
\end{equation}
which consists of two ones on positions $i_0$ and $i_1$. We thus have $h(v,\phi)=1$, meaning that for $v$ the bipartition represented by $\phi$ is odd. Clearly, following the above reasoning, for any nontrivial bipartition $\phi$
such a vector can be constructed; let us denote it $v_{\phi}\in F_N$. In other words, for any nontrivial bipartition there exist a vector on which this bipartition is odd. 

Now, the fact that $\dim K(\mathbb{S})=N-1$ implies that the vectors $v_{i,j}$ constructed from the generators $G_i$ span $W_N$. Let us then select those $v_{i,j}$ that form a basis in $W_N$ and denote them $u_1,\ldots,u_{N-1}$; recall that in general there might be more vectors $v_{i,j}$ than $\dim W_N=N-1$. Clearly, any vector $v_{\phi}$ can be decomposed in that basis, that is,
\begin{equation}
    v_{\phi}=\sum_{i=1}^{N-1}\alpha_i u_i,
\end{equation}
where $\alpha_i\in\{0,1\}$ for all $i$. As mentioned above $h(v_{\phi},\phi)=1$, which due to the above decomposition implies that for some $i$ such that $\alpha_i\neq 0$,
\begin{equation}
h(\phi,u_i)=1.
\end{equation}
Consequently, any nontrivial bipartition $\phi$ is odd on at least one vector $v_{i,j}$ and thus it is odd on $K(\mathbb{S})$. Lemma \ref{Lem0} implies then that the subspace $V$ is genuinely entangled.
\end{proof}

Thus, this theorem tell us that a stabilizer subspace
is genuinely entangled if, and only if the corresponding vectors $v_{i,j}$ span $W_N$, or, equivalently, the subspace $K(\mathbb{S})$ spanned by them is isomorphic to $W_N$.

Let us illustrate the power of this criterion
by applying it to show that the examples of stabilizer states (graph states) and subspaces (5-qubit code) introduced in the preliminary section (Sec. \ref{prelim}) are indeed genuinely entangled.

\textit{Graph states.} First, we consider the graph states. 
Let us choose a particular graph $G=(V,E)$ 
and the corresponding graph state $\ket{\psi_G}$
stabilized by the stabilizing operators given in 
Eq. (\ref{StabGrSt}). Now, for a given pair  
$G_i$ and $G_j$ with $i\neq j$, the corresponding vector $v_{i,j}$ will have either two (if the vertices $i,j$ are connected) or no ones (if the they are unconnected). Let us then consider all vectors obtained from the stabilizing operators $G_i$ that have two ones and assume that they do not span $W_N$. This implies that there
exists a nontrivial even bipartition $Q|\overline{Q}$ of $v_{ij}$, i.e., all vectors $v_{ij}$ can be written as direct sums
\begin{equation}
    v_{i,j}=v_{i,j}^Q\oplus v_{i,j}^{\overline{Q}},
\end{equation}
where all $v_{i,j}^Q$ and $v_{i,j}^{\overline{Q}}$ have even number of ones or no ones. However, if such a bipartition exists it follows that the graph $G$ consists of two unconnected subgraphs, which contradicts that we consider only connected graphs. Thus, the graph state $\ket{\psi_G}$ is genuinely entangled. 

\textit{The five-qubit code.} Let us now apply our criterion to the five-qubit code; the generators
are given in Eq. (\ref{Stab5Qubit}). One finds that the vectors $v_{i,j}$
that are obtained from them are of the form
\begin{eqnarray}
v_{1,2}&=&(0,1,0,1,0),\qquad v_{1,3}=(0,0,1,1,0),\nonumber\\
v_{1,4}&=&(1,1,0,0,0),\qquad v_{2,3}=(0,0,1,0,1),\nonumber\\
v_{2,4}&=&(0,0,0,1,1),\qquad v_{3,4}=(1,0,0,1,0).
\end{eqnarray}
We thus have six vectors with, each having two ones. It is not difficult to see that $v_{1,2}$, $v_{1,3}$, $v_{1,4}$ and $v_{2,4}$ are linearly independent and therefore they span $W_4$, which implies the corresponding subspace to be GME.

\subsection{Construction of GME stabilizer subspaces for any number of qubits}
\label{SecIIIB}

Interestingly, the above formalism, in particular Theorem \ref{Crit2} not only serve as a efficient method of checking genuine entanglement of a stabilizer subspace, but can also be employed to construct GME stabilizer subspaces for any number of qubits. More precisely, given a basis spanning $W_N$ like the one in Eq. (\ref{vectorsW}) one can try to find stabilizing operators giving rise to this basis. 
In other words, any abelian stabilizer $\mathbb{S}=\langle G_1,\ldots,G_k\rangle$ that does not contain $-\mathbbm{1}^{\otimes N}$ and gives rise to these vectors, stabilizes a nontrivial genuinely entangled subspace. 

Let us illustrate the above idea with two examples of nontrivial stabilizers giving rise to vectors (\ref{vectorsW}). The first one corresponds to the well-known $N$-qubit GHZ state, whereas the second one gives rise to a class of subspaces whose dimension grown exponentially with $N$

\textit{Construction 1.} A very simple construction of generators that giving rise to the vectors (\ref{vectorsW}) is the following one
\begin{equation}
G_0=X_1\ldots X_N
\end{equation}
and
\begin{equation}
G_i=Z_iZ_{i+1}
\end{equation}
with $1=2,\ldots,N-1$. Indeed, a pair $G_0,G_i$ with $i=1,\ldots,N-1$
gives rise to the vector $w_i$. Moreover, the operators $G_i$ commute because each $G_i$ with $i=1,\ldots,N-1$ contains only two $Z$ operators. It is also fairly easy to see that one cannot obtain $-\mathbbm{1}^{\otimes N}$ by taking their products, and consequently, 
they give rise to a nontrivial stabilizer. Due to the fact that they are all independent, the dimension of the stabilized subspace is one, and it corresponds to the $N$-qubit GHZ state
\begin{equation}
    \ket{\mathrm{GHZ}_N}=\frac{1}{\sqrt{2}}(\ket{0}^{\otimes N}+\ket{1}^{\otimes N}).
\end{equation}

\textit{Construction 2.} Let us now present a bit less straightforward construction of a GME stabilizer subspace whose dimension is higher than one and grown with the number of parties $N$. For clarity we assume
$N$ to be even, however, a similar construction can also be given for odd $N$.

Let us begin with the simplest cases 
of $N=4$ and $N=6$. For $N=4$ we take the following operators
\begin{eqnarray}\label{Gexample0}
   \mathcal{G}_1&=&X_1X_2X_3X_4,\nonumber\\
    \mathcal{G}_2&=&Z_1Z_2X_3,\nonumber\\
    \mathcal{G}_3&=&X_2Z_3Z_4.
\end{eqnarray}
These operators commute and it is not possible to 
generate $-\mathbbm{1}_2^{\otimes 3}$ by taking their products because 
to obtain $-\mathbbm{1}_2^{\otimes 3}$ we need a product of at least four different generators, whereas here we have only three. Now, one finds that the vectors 
(\ref{vectors}) that can be associated to them are exactly
\begin{eqnarray}
v_{1,2}&=&(1,1,0,0),\nonumber\\
v_{2,3}&=&(0,1,1,0),\nonumber\\
v_{1,3}&=&(0,0,1,1),
\end{eqnarray}
that is, they span $W_3$, and therefore the two-dimensional subspace
$\mathcal{V}_4$ stabilized by (\ref{Gexample0}) is genuinely entangled. 

For $N=6$ we consider the following four operators
\begin{eqnarray}\label{Gexample1}
   \mathcal{G}_1 &=& X_1X_2X_3X_4X_5 X_6,\nonumber\\
\mathcal{G}_2 &=& Z_1Z_2X_3,\nonumber\\
    \mathcal{G}_3&=& X_2Z_3Z_4X_5,\nonumber\\
    \mathcal{G}_4&=& X_4Z_5Z_6.
\end{eqnarray}
Again, they mutually commute and do not generate $-\mathbbm{1}^{\otimes 6}_2$, and thus give rise to a stabilizer that stabilizes a nontrivial subspace $\mathcal{V}_6$ of dimension four. These generators give rise to six vectors $v_{i,j}$ among which five are of the form (\ref{vectors}), or, equivalently, $\dim \mathrm{span}\{v_{i,j}\}=5$, and thus the stabilized subspace $\mathcal{V}_6$ is GME. 

Given these two particular cases it is not difficult to construct the
generators $G_i$ for any even $N$ giving rise to the vectors (\ref{vectorsW}). Indeed, let us consider the following operators
\begin{eqnarray}\label{GenII}
\mathcal{G}_1&=&X_1X_2\ldots X_N,\nonumber\\
\mathcal{G}_2&=&Z_1Z_2X_3,\nonumber\\
\mathcal{G}_i&=&X_{2i-4}Z_{2i-3}Z_{2i-2}X_{2i-1}
\end{eqnarray}

with $i=3,\ldots,N/2$, and
\begin{eqnarray}\label{GenIII}
\mathcal{G}_{N/2+1}&=&X_{N-2}Z_{N-1}Z_N.
\end{eqnarray}
One realizes that they all mutually commute:
$\mathcal{G}_{1}$ commutes with any $\mathcal{G}_i$ for $i=2,\ldots,N/2+1$
because the latter consists of two $Z$ operators and one or two $X$ operators, whereas any pair of operators $\mathcal{G}_i$, $\mathcal{G}_j$ for $i,j=2,\ldots,N/2+1$ commutes because they overlap nontrivially on two positions, in which case the local Pauli operators anticommute on these two positions, or they do not have nontrivial overlap. Then, the subgroup generated by $\mathcal{G}_i$ does not contain $-\mathbbm{1}^{\otimes N}_2$ due to the fact that the largest 
overlap of any $k\geq 4$ different generators is three, while one needs
$XZXZ$ to create $-\mathbbm{1}_2$ at a single site.

Now, one finds that these generators are independent and there are $N/2+1$ of them, and therefore they stabilize a subspace $\mathcal{V}_N$ of dimension $\dim \mathcal{V}_{N}=2^{N-k}=2^{N-N/2-1}=2^{N/2-1}$. 
Moreover, this subspace is GME for any $N$ because the vectors 
$v_{i,j}$ these generators give rise to contain those listed in Eq.  (\ref{vectorsW}). Yet, $2^{N/2-1}$ is not the maximal dimension of a GME subspace that can be achieved within the stabilizer formalism, and, as we will see in the next section, larger subspaces can be constructed and we provide an example of such maximally-dimensional subspace. 

Let us conclude this section by noting that one can slightly modify 
the generators $\mathcal{G}_i$ for any even $N\geq 6$ to make them cyclic
\begin{eqnarray}\label{wzory}
\widetilde{\mathcal{G}}_1&=&X_1\ldots X_N\nonumber\\
\widetilde{\mathcal{G}}_i&=&X_{2i-4}Z_{2i-3}Z_{2i-2}X_{2i-1}
\end{eqnarray}
with $i=2,\ldots,N/2+1$, where we use the convention that $X_{0}\equiv X_N$.

\section{Genuinely entangled stabilizer subspaces of maximal dimension}

In the preceding section we presented a nontrivial example of a stabilizer GME subspace whose scales exponentially with the number of qubits $N$. Our aim here is to provide another example of 
a GME stabilizer subspace whose dimension for any $N$ is maximal
achievable within the stabilizer formalism.

Let us first find an upper bound on the 
maximal dimension of the GME stabilizer subspaces (see also \cite{Waegell}).
\begin{thm}
For any nontrivial stabilizer $\mathbb{S}\subset \mathbb{G}_N$ for which stabilizer subspace $V(\mathbb{S})$ is genuinely entangled, the dimension of $V(\mathbb{S})$ is bounded as
\begin{equation}\label{BoundDim}
    \dim V\leq 2^{N-k_{\min}(N)},
\end{equation}
where
\begin{equation}\label{Nmin}
    k_{\min}(N)=\left\lceil \frac{1+\sqrt{8N-7}}{2}\right\rceil.
\end{equation}
\end{thm}
\begin{proof}
It follows from Theorem \ref{Crit2} that we need at least $N-1$ vectors
$v_{i,j}$ from $W_N$ for a stabilizer subspace to be GME. Given that 
the stabilizer $\mathbb{S}_N$ is generated by $k$ generators, there are
$k(k-1)/2$ vectors $v_{i,j}$. We thus have an inequality relating $k$ and $N$:
\begin{equation}
    N-1\leq \frac{1}{2}k(k-1)
\end{equation}
that must be satisfied in order for the subspace $V(\mathbb{S})$
to be GME. The smallest $k$ obeying this inequality is given by (\ref{Nmin}). 
\end{proof}

Let us now show that the bound (\ref{BoundDim}) is tight 
by providing a particular GME stabilizer subspace for which 
(\ref{BoundDim}) is saturated. To this end, to simplify the notation, let us first denote $k\equiv k_{\min}(N)$ and 
\begin{equation}\label{Pi}
P_{i}=\left\{ 
\begin{array}{ll}
X,& \textrm{ for odd\;} i, \\[1ex]
Z,& \textrm{ for even\;}i.
\end{array}
\right.
\end{equation}
Consider then the following operators
\begin{eqnarray}\label{eq_stab}
G_{i}^{\max}&=&\mathbbm{1}_2^{\otimes\frac{1}{2}(i-1)(i-2)}\otimes P_{i}^{\otimes(i+1)}\nonumber\\
&&\bigotimes_{q=i+2}^{q_{\max}(i)}\left[\mathbbm{1}_2^{\otimes(q-3)}\otimes P_{q+1}^{\otimes 2}\right]\otimes\mathbbm{1}_2^{\otimes[N-\gamma(i)]},\nonumber\\
G_{k-1}^{\max}&=&\mathbbm{1}_2^{\otimes\frac{1}{2}(k-2)(k-3)}\otimes P_{k-1}^{\otimes k}\otimes\mathbbm{1}_2^{\otimes\left[N-\frac{1}{2}(k-2)(k-3)-k\right]},\nonumber\\
G_{k}^{\max}&=&\mathbbm{1}_2^{\otimes\frac{1}{2}(k-1)(k-2)}\otimes P_{k}^{\otimes\left[N-\frac{1}{2}(k-1)(k-2)\right]},
\end{eqnarray}
where $i=1,\dots,k-2$ and $q_{\max}(i)\coloneqq k_{\min}(N+i)-1$ with $k_{\min}(N+i)$ being (\ref{Nmin})
for $N+i$ qubits, and
\begin{equation}
\gamma(i)=\frac{1}{2}\left[q_{\max}^{2}(i)-q_{\max}(i)-2i+4\right].
\end{equation}

We can now state the following fact, 
whose proof is quite technical and therefore it is deferred to Appendix \ref{AppA}. 
\begin{fakt}\label{fakt2}
The operators $G_i^{\max}$ defined in Eqs. (\ref{eq_stab}) mutually commute and do not generate $-\mathbbm{1}_2^{\otimes N}$, and therefore 
stabilize a nontrivial subspace $V_N^{\max}$ of dimension $\dim V_{N}^{\max}=2^{N-k_{\min}(N)}$.
\end{fakt}

It is interesting to compare the maximal dimension of GME stabilizer subspaces (\ref{BoundDim}), which for sufficiently large $N$ scales as $2^{N-\sqrt{2N}}$, with the maximal dimension of any GME subspace is $(\mathbbm{C}^2)^{\otimes N}$ given by the following simple formula \cite{Demianowicz},
\begin{equation}
    D_{\max}^{\mathrm{GES}}=2^{N-1}-1.
\end{equation}
Thus, as expected, the maximal dimension of any GME subspace is bigger for any $N$ than that of GME stabilizer subspaces, and, moreover, gap between them grows with $N$ as their ratio scales as $2^{\sqrt{2N}}$.

For illustrative purposes let us finally discuss 
our construction for two particular cases of 
$N=3$ and $N=4$ (for $N=2$ our construction reproduces stabilizing operators of the maximally entangled state of two qubits). 

For $N=3$, Eqs. (\ref{eq_stab}) give
\begin{equation}\label{N3}
G_1^{\max}=X_1X_2,\quad G_2^{\max}=Z_1Z_2Z_3,\quad G_3^{\max}=X_2X_3.
\end{equation}
It is direct to see that they are independent and generate a nontrivial stabilizer $\mathbb{S}_{3}^{\max}=\langle G^{\max}_1,G^{\max}_2,G^{\max}_3\rangle$, implying that 
$\mathbb{S}^{\max}_3$ stabilizes a one-dimensional subspace
spanned by
\begin{eqnarray}
\ket{\psi^{\max}_3}&=&\frac{1}{2}(\ket{000}+\ket{110}+\ket{011}+\ket{101})\nonumber\\
&=&\frac{1}{\sqrt{2}}(\ket{\psi_{+}}\ket{0}+\ket{\phi_{+}}\ket{1}),
\end{eqnarray}
where $\ket{\psi_{\pm}}=(\ket{00}\pm\ket{11})/\sqrt{2}$
and $\ket{\phi_{\pm}}=(\ket{01}\pm\ket{10})/\sqrt{2}$ is the two-qubit Bell basis. 

Moreover, the corresponding subspace $K(\mathbb{S}_{3}^{\max})$ is two-dimensional
as the vectors $v_{i,j}$ associated to the pairs $G_1^{\max},G_2^{\max}$ and $G_2^{\max},G_3^{\max}$ are $(1,1,0)$ and $(0,1,1)$, respectively. Consequently, $\ket{\psi^{\max}_3}$ is genuinely entangled.

Let us now move to a bit more complicated case of $N=4$. 
The stabilizing operators are given by 
\begin{equation}
    G_1=X_1X_2Z_3Z_4,\quad G_2=Z_1Z_2Z_3,\quad G_3=X_2X_3X_4
\end{equation}
Again, they mutually commute and one cannot generate $-\mathbbm{1}^{\otimes 4}_2$ by taking their products. 
Hence, they stabilize a two-dimensional subspace $V^{\max}_4$
spanned by the four-qubit states that can be written in the following way
\begin{equation}
\ket{\psi_1}=\frac{1}{2}(\ket{\psi_+}\ket{00}+\ket{\psi_-}\ket{01}+\ket{\phi_-}\ket{10}+\ket{\phi_+}\ket{11})
\end{equation}
and
\begin{equation}
\ket{\psi_2}=\frac{1}{2}(\ket{\psi_+}\ket{00}-\ket{\psi_-}\ket{01}-\ket{\phi_-}\ket{10}+\ket{\phi_+}\ket{11}).  
\end{equation}

One finally checks that the subspace $K(\mathbb{S}^{\max}_{4})$ associated to $G_i$ $(i=1,2,3)$ is three dimensional and therefore $V^{\max}_{4}$ is genuinely entangled.

\section{Bell inequalities maximally violated by stabilizer subspaces}

Interestingly, taking advantage of the stabilizer formalism
we can construct multipartite Bell inequalities that are maximally
violated by any state from the maximally-dimensional subspaces
introduced in the preceding section (see also Ref. \cite{Waegell,Koh} for earlier constructions of Bell inequalities maximally violated by entangled subspaces). Our construction 
builds on a general approach for constructing CHSH-like Bell inequalities within the stabilizer formalism introduced in Refs. \cite{graphstates,subspaces}. (In Appendix \ref{AppB} we also provide an analogous construction for Bell inequalities maximally violated by the subspaces 
given in Sec. \ref{SecIIIB}.)

Let us consider again the parties $P_i$ and denote by 
$A^{(i)}_{x_i}$ the observables measured by them on their shares
of some state $\rho$. The departure point of our construction of Bell inequalities are the stabilizing operators (\ref{eq_stab}). For pedagogical reasons we begin, however, with the simplest case $N=3$, for which they are stated explicitly in Eq. (\ref{N3}).

Now, we assign the observables $A_{x_i}^{(i)}$ or their combinations to all Pauli matrices $X$ and $Z$ appearing in 
$G_i^{\max}$ $(i=1,2,3)$.
Precisely, for the first observer we make the following
assignments
\begin{equation}\label{assign1}
X\to A_{0}^{(1)}+A_1^{(1)},\qquad
Z\to A_{0}^{(1)}-A_1^{(1)},
\end{equation}
whereas for the remaining two observers $P_i$ $(i=2,3)$:
\begin{equation}\label{assign2}
X\to A_{0}^{(i)}, \qquad Z\to A_{1}^{(i)}.
\end{equation}

Then, using these assignments, we associate an expectation value to each generator $G_i^{\max}$. Precisely, 
\begin{eqnarray}
    G_1^{\max}&\to& \langle (A_{0}^{(1)}+A_1^{(1)})A_0^{(2)}\rangle,\nonumber\\
    G_2^{\max}&\to& \langle (A_{0}^{(1)}-A_1^{(1)})A_1^{(2)}A_1^{(3)}\rangle
\end{eqnarray}
and
\begin{equation}
    G_3^{\max}\to \langle A_0^{(2)}A_0^{(3)}\rangle\nonumber\\
\end{equation}
Adding the obtained expectation values we finally arrive at the  following Bell inequality
\begin{eqnarray}\label{BellN3}
    I_{3}^{\max}&:=&\langle (A_{0}^{(1)}+A_1^{(1)})A_0^{(2)}\rangle+
    \langle (A_{0}^{(1)}-A_1^{(1)})A_1^{(2)}A_1^{(3)}\rangle\nonumber\\
    &&+\langle A_0^{(2)}A_0^{(3)}\rangle\leq 3,
\end{eqnarray}
whose classical bound $\beta_{C}^3=3$ was found by optimizing $I_3^{\max}$ over all local deterministic correlations [cf. Eq. (\ref{locdet}) and the text that follows].

Interestingly, the maximal quantum value of (\ref{BellN3}) can also be found analytically and is given by $\beta_Q^3=2\sqrt{2}+1>\beta_C^3$. Precisely, observing that for any choice of dichotomic observables $A_{x_i}^{(i)}$, 
the "shifted Bell operator" $(2\sqrt{2}+1)\mathbbm{1}-\mathcal{B}$ decomposes into the following sum of positive operators
\begin{eqnarray}
    (2\sqrt{2}+1)\mathbbm{1}-\mathcal{B}&=&\frac{1}{\sqrt{2}}
    \left[\mathbbm{1}-\frac{1}{\sqrt{2}}(A_{0}^{(1)}+A_1^{(1)})\right]^2\nonumber\\
    &&+\frac{1}{\sqrt{2}}
    \left[\mathbbm{1}-\frac{1}{\sqrt{2}}(A_{0}^{(1)}-A_1^{(1)})\right]^2\nonumber\\
    &&+\frac{1}{2}\left[\mathbbm{1}-A_0^{(2)}A_0^{(3)}\right]^2,
\end{eqnarray}
one infers that the maximal quantum value of $I_{3}^{\max}$ is upper bounded as $\beta_Q^3\leq 2\sqrt{2}+1$. To show that this is in fact an equality it suffices to notice that for the following choice of observables 
\begin{equation}\label{ChoiceN3}
A_0^{(1)}=\frac{1}{\sqrt{2}}(X+Z),\qquad 
A_1^{(1)}=\frac{1}{\sqrt{2}}(X-Z),\qquad 
\end{equation}
for the first party and $A_0^{(i)}=X$ and $A_1^{(i)}=Z$
for $i=2,3$, the Bell operator becomes a weighted sum of
the generators, that is, $\mathcal{B}=\sqrt{2}\,(G_1^{\max}+G_2^{\max})+G_3^{\max}$.
Its expectation value in the state $\ket{\psi_{3}^{\max}}$ stabilized by $G_i^{\max}$ is exactly $2\sqrt{2}+1$, which implies that $\beta_Q^3=2\sqrt{2}+1$.

\subsection{The construction for arbitrary $N$}
\label{ConstructionN}

The above construction generalizes directly to arbitrary number of observers $N$. Let us go back to the stabilizing operators (\ref{eq_stab}) and follow
the same steps as above. That is, for the first observer we make the assignments (\ref{assign1}), whereas
for the remaining ones we substitute 
\begin{equation}\label{assign22}
    X\to A_0^{(i)},\qquad Z\to A_1^{(i)}
\end{equation}
for $i=2,\ldots,N$. We then associate 
expectation values to all $G_i^{\max}$ in which the Pauli matrices
are replaced by the observables or their combinations
according to these assignments. Noting that $G_1$ and $G_2$ are the only generators that contain nontrivial matrices $X$ and $Z$ at the first site, we have
\begin{eqnarray}\label{G1max}
G_1^{\max}&\to& \left\langle(A_0^{(1)}+A_1^{(1)})A_0^{(2)}\right.
\nonumber\\
&&\left.\times\prod_{q=1}^{q_{\max}(1)-2}A_{q}^{((q+1)(q+2)/2)}A_{q}^{((q+1)(q+2)/2+1)}\right\rangle\nonumber\\
\end{eqnarray}
and
\begin{eqnarray}\label{G2max}
\hspace{-1cm}G_2^{\max}&\to& \left\langle(A_0^{(1)}-A_1^{(1)})A_1^{(2)}A_1^{(3)}
\right.
\nonumber\\
&&\left.\times\prod_{q=2}^{q_{\max}(2)-3}A_{q}^{((q+1)(q+2)/2-1)}A_{q}^{((q+1)(q+2)/2)}\right\rangle,\nonumber\\
\end{eqnarray}
where here and below we use the convention  
that $A_{q}^{(i)}$ stands for $A_0^{(i)}$ or $A_1^{(i)}$ 
if $q$ is even or odd, respectively. 
For the other generators we have
\begin{eqnarray}
G_{i}^{\max}&\to& \left\langle \prod_{j=1}^{i+1}A_{i+1}^{((i-1)(i-2)/2+j)}
\right.
\nonumber\\
&&\left.\times\prod_{q=i}^{q_{\max}(i)-(i+1)}A_{q}^{((q+1)(q+2)/2-i+1)}A_{q}^{((q+1)(q+2)/2-i+2)}\right\rangle,\nonumber\\
\end{eqnarray}
Finally, for the last two generators we have
\begin{equation}
G_{k-1}^{\max}\to \left\langle A_{k}^{((k-2)(k-3)/2+1)}\ldots A_{k}^{((k-2)(k-3)/2+k)}\right\rangle\nonumber\\
\end{equation}
and
\begin{equation}
G_{k}^{\max}\to \left\langle A_{k+1}^{((k-1)(k-2)/2+1)}\ldots A_{k+1}^{(N)}\right\rangle.
\end{equation}
\begin{widetext}

By adding the above expectation values we arrive at the
following Bell expression
\begin{eqnarray}\label{NierMax}
I_{\max}^N:&=&\left\langle(A_0^{(1)}+A_1^{(1)})A_0^{(2)}
\prod_{q=1}^{q_{\max}(1)-2}A_{q}^{((q+1)(q+2)/2)}A_{q}^{((q+1)(q+2)/2+1)}\right\rangle\nonumber\\
&&+\left\langle(A_0^{(1)}-A_1^{(1)})A_1^{(2)}A_1^{(3)}
\prod_{q=2}^{q_{\max}(2)-3}A_{q}^{((q+1)(q+2)/2-1)}A_{q}^{((q+1)(q+2)/2)}\right\rangle\nonumber\\
&&+\sum_{i=3}^{k-2}\left\langle \prod_{j=1}^{i+1}A_{i+1}^{((i-1)(i-2)/2+j)}\prod_{q=i}^{q_{\max}(i)-(i+1)}A_{q}^{((q+1)(q+2)/2-i+1)}A_{q}^{((q+1)(q+2)/2-i+2)}\right\rangle\nonumber\\
&&+\left\langle A_{k}^{((k-2)(k-3)/2+1)}\ldots A_{k}^{((k-2)(k-3)/2+k)}\right\rangle+\left\langle A_{k+1}^{((k-1)(k-2)/2+1)}\ldots A_{k+1}^{(N)}\right\rangle
\leq k_{\min}(N).
\end{eqnarray}
\end{widetext}
The maximal classical value of these inequalities has been found by noting that for the local deterministic correlations, that is, those for which all expectation values are product
$\langle A_{x_{i_1}}^{(i_1)} \ldots A_{x_{i_k}}^{(i_k)}\rangle =\langle A_{x_{i_1}}^{(i_1)}\rangle\ldots \langle A_{x_{i_k}}^{(i_k)}\rangle$ and $\langle A_{x_{i_k}}^{(i_k)}\rangle=\pm1 $, all expectation values
in the third and the fourth line of (\ref{NierMax})
can take the maximal value one, whereas only one of the first two terms can be nonzero and takes the maximal value two.

Remarkably, the maximal quantum value of $I_N^{\max}$ can 
also be found analytically for any $N$. 
\begin{fakt}
The maximal quantum value of $I_N^{\max}$ is 
$\beta_Q^N=2(\sqrt{2}-1)+k_{\min}(N)$ and it is achieved by 
any pure state from $V_N^{\max}$. 
\end{fakt}
\begin{proof}
The proof goes along the same lines as that for $N=3$. 
First, one finds by a direct check that the following decomposition
into a sum of positive operators is true, 
\begin{eqnarray}\label{eq_sumofsquares}
&&\hspace{-0.3cm}[2(\sqrt{2}-1)+k_{\min}(N)]\mathbbm{1}_2-\mathcal{B}_{N}
\nonumber\\
&&=\frac{1}{\sqrt{2}}(\mathbbm{1}_2-\tilde{G}_{1})^{2}\nonumber+\frac{1}{\sqrt{2}}(\mathbbm{1}_2-\tilde{G}_{2})^{2}
+\frac{1}{2}\sum_{i=3}^{k}(\mathbbm{1}_2-\tilde{G}_{i})^{2},
\end{eqnarray}
where $\mathcal{B}_N$ is a Bell operator associated to the 
Bell expression $I_N^{\max}$ with arbitrary observables $A_{x_i}^{(i)}$, and $\tilde{G}_{i}$ are constructed from 
the generators $G_i^{\max}$ by replacing the $X$ and $Z$ operators
by the observables $A_{x_i}^{(i)}$. This decomposition implies that $\beta_Q^N\leq 2(\sqrt{2}-1)+k_{\min}(N)$. 

In full analogy to the $N=3$ case we then let the parties measure the following observables:
for the first party we take those defined in Eq. (\ref{ChoiceN3}), 
whereas for the remaining ones $A_0^{(i)}=X$ and $A_1^{(i)}=Z$
with $i=2,\ldots,N$. For this particular choice the Bell operator becomes
\begin{equation}
    \mathcal{B}_N=\sqrt{2}(G_1^{\max}+G_2^{\max})+\sum_{i=3}^{k_{\min}(N)}G_i^{\max},
\end{equation}
and so it is a simple combination of the stabilizing operators $G_i^{\max}$ that stabilize the subspace $V_N^{\max}$. Consequently, 
for any pure state $\ket{\psi}\in V_N^{\max}$, one has 
$\langle\psi|\mathcal{B}_N|\psi\rangle=2(\sqrt{2}-1)+k_{\min}(N)$, 
which is what we wanted to achieve.
\end{proof}
Two comments are in order. First, we stress that the inequalities (\ref{NierMax}) are nontrivial for any number of parties as $\beta_Q>\beta_C$ for any $N\geq 2$. Second, it is worth pointing out that they are not only maximally violated by any pure state from $V_N^{\max}$ but also by any mixed state supported on $V_N^{\max}$.

In the next section we demonstrate that 
they can be used to self-test those maximally-dimensional subspaces.

\section{Self-testing the stabilizer subspaces of maximal dimension}

Let us first recall the definition of subspace self-testing introduced in Ref. \cite{subspaces}. Consider an entangled subspace $V$ of the $N$-qubit Hilbert space $\mathcal{H}_{P'}=(\mathbbm{C}^2)^{\otimes N}$ spanned by 
some orthogonal vectors $\ket{\psi_i}\in\mathcal{H}_{P'}$. Imagine then that the $N$ parties perform a Bell experiment on some $N$-partite (generally mixed) state acting on $\mathcal{H}_P$ whose dimension we assume to be finite but unknown and observe correlations $\mathcal{P}$.
Denoting $\ket{\phi_{PE}}\in \mathcal{H}_{PE}$ a purification of 
$\rho_P$ to a larger Hilbert space $\mathcal{H}_{PE}=\mathcal{H}_P\otimes\mathcal{H}_E$, we have the following definition \cite{subspaces}.
\begin{defin}
We say that the correlation $\mathcal{P}$ self-test the 
entangled subspace $V$ if for any pure state $\ket{\phi_{PE}}$
compatible with $\mathcal{P}$ one can deduce that:
(i) every local Hilbert space $\mathcal{H}_{P_i}=\mathbbm{C}^2\otimes \mathcal{H}_{P_i}''$
and (ii) there exist a local unitary operation $U_P=U_1\otimes\ldots\otimes U_N$ acting on $\mathcal{H}_P$ such that 
\begin{equation}
    (U_P\otimes\mathbbm{1}_E)\ket{\phi}_{PE}=\sum_{i}\alpha_i\ket{\psi_i}_{P'}\otimes\ket{\xi_i}_{P''E},
\end{equation}
where $\alpha_i\in\mathbb{R}$ are some nonnegative coefficients such that
$\sum_{i}\alpha_i^2=1$, whereas $\ket{\xi_i}_{P''E}$ are pure states from $\mathcal{H}_{P''E}$. 
\end{defin}

Let us now prove that the maximal quantum violation of 
the Bell inequalities introduced in the preceding section
self-tests the subspaces $V_{N}^{\max}$ for any $N$.
\begin{thm}
Any correlations maximally violating the Bell inequality 
(\ref{NierMax}) for a given $N$ self-tests the
subspace $V_{N}^{\max}$.
\end{thm}

Before proving this theorem we recall a well-known fact that will be highly useful for our proof. 
\begin{lem}\label{lem_selftest}
\cite{PR,Jed5} Let $\tilde{X}$ and $\tilde{Z}$ be Hermitian operators acting on $\mathcal{H}$ such that $\dim\mathcal{H}=d<\infty$. If the following conditions are met
\begin{enumerate}
    \item $\tilde{X}^{2}=\tilde{Z}^{2}=\mathbbm{1}_{\mathcal{H}}$,
    \item $\{\tilde{X},\tilde{Z}\}=0$,
\end{enumerate}
then we have $\mathcal{H}=\mathbbm{C}^{2}\otimes\mathbbm{C}^{d'}$, or, in other words, $d=2d'$ for some positive integer $d'$, and there exists a local unitary operation $U$ such that
\begin{equation*}
U\tilde{X}U^{\dagger} = X\otimes \mathbbm{1}_{d'},\quad U\tilde{Z}U^{\dagger} = Z\otimes \mathbbm{1}_{d'},
\end{equation*}
where $X$ and $Z$ are the Pauli matrices.
\end{lem}

\begin{proof}
Let us denote by $A_{x_i}^{(i)}$ with $x_i=0,1$ and $i=1,\ldots,N$,
and $\ket{\psi}\in \mathcal{H}_{PE}$ denote quantum observables
and quantum state achieving the maximal quantum value of
$I_N^{\max}$. Without any loss of generality we can assume
that the reduced density matrices $\rho_{P_i}$ are of full rank.

From the sum-of-squares decomposition we infer that 
\begin{equation}\label{conds}
    \tilde{G}_i\ket{\tilde{\phi}}=\ket{\tilde{\phi}},
\end{equation}
where we recall that $\tilde{G}_i$ are quantum operators acting on $\mathcal{H}_P$ constructed from the generators
$G_i$ by making the substitutions (\ref{assign1}) and (\ref{assign22}); for clarity in Eq. (\ref{conds}) and in what follows we omit the identity acting on $E$. To simplify the mathematical considerations let us denote 
\begin{equation}\label{eq_substit_real}
\tilde{X}^{(1)} =\frac{A^{(1)}_{0}+A^{(1)}_{1}}{\sqrt{2}},\qquad \tilde{Z}^{(1)} =\frac{A^{(1)}_{0}-A^{(1)}_{1}}{\sqrt{2}}
\end{equation}
and
\begin{equation}
\tilde{X}^{(i)} = A^{(i)}_{0},\qquad \tilde{Z}^{(i)}=A^{(i)}_{1}.
\end{equation}

Let us now show that all operators $\tilde{X}^{(i)}$ and $\tilde{Z}^{(i)}$ satisfy both conditions of 
Lemma \ref{lem_selftest}. We begin with the first condition.
Due to the fact that $\tilde{X}^{(i)}$ and $\tilde{Z}^{(i)}$
for $i=2,\ldots,N$ are Hermitian operators with eigenvalues $\pm1$, we directly see
that 
\begin{equation}
\left[\tilde{X}^{(i)}\right]^2=\left[\tilde{Z}^{(i)}\right]^2=\mathbbm{1}
\end{equation}
To prove that the same holds for 
$\tilde{X}^{(1)}$ and $\tilde{Z}^{(1)}$,
let us consider (\ref{conds}) for $i=1$ which 
can be stated as
\begin{equation*}
\tilde{X}^{(1)}\otimes\tilde{G}^{(Q_{1})}_{1}\ket{\tilde{\phi}}=\ket{\tilde{\phi}},
\end{equation*}
where $\tilde{G}^{(Q_{1})}_{1}$ denotes all elements of 
$\tilde{G}^{(1)}$ apart from $\tilde{X}^{(1)}$. By applying
$\tilde{X}^{(1)}\otimes\tilde{G}^{(Q_{1})}_{1}$ to both sides of
the above identity and then exploiting the fact that 
$\tilde{G}^{(Q_{1})}_{1}$ squares to the identity, we 
arrive at 
\begin{equation*}
[\tilde{X}^{(1)}]^2\ket{\tilde{\phi}}=\ket{\tilde{\phi}},
\end{equation*}
which is equivalent to $[\tilde{X}^{(1)}]^2=\mathbbm{1}$; recall that the reduced density matrix $\rho_{P_1}$ acting on $\mathcal{H}_{P_1}$ is of full rank. Analogously, we can use the condition (\ref{conds}) for $i=2$ to conclude that $[\tilde{Z}^{(1)}]^2=\mathbbm{1}$. 

We can now proceed with the second condition of Lemma \ref{lem_selftest}. From Eq. (\ref{eq_substit_real})
one directly realizes that 
\begin{equation}
    \{\tilde{X}^{(1)},\tilde{Z}^{(1)}\}=0
\end{equation}
To prove the same for the remaining sites we need to 
employ the conditions (\ref{conds}). Let us thus consider those for $i=1$ and $i=2$ and write them as
\begin{equation}
    \tilde{X}^{(1)}\tilde{X}^{(2)}\tilde{G}_1^{(Q_2)}\ket{\tilde{\phi}}=\ket{\tilde{\phi}}
\end{equation}
and
\begin{equation}
    \tilde{Z}^{(1)}\tilde{Z}^{(2)}\tilde{G}_2^{(Q_2)}\ket{\tilde{\phi}}=\ket{\tilde{\phi}},
\end{equation}
where $\tilde{G}_i^{(Q_2)}$ are the parts of  $\tilde{G}_2$ corresponding to the subset $Q_2=\{3,\ldots,N\}$. From these relations we obtain that
\begin{eqnarray}
&&\hspace{-0.5cm}\{\tilde{X}^{(2)},\tilde{Z}^{(2)}\}\ket{\tilde{\phi}}\nonumber\\
&&=(\tilde{Z}^{(1)}\tilde{X}^{(1)}\tilde{G}^{(Q_{2})}_{2}\tilde{G}^{(Q_{2})}_{1}+\tilde{X}^{(1)}\tilde{Z}^{(1)}\tilde{G}^{(Q_{2})}_{1}\tilde{G}^{(Q_{2})}_{2})\ket{\tilde{\phi}}\nonumber\\
&&=\{\tilde{X}^{(1)},\tilde{Z}^{(1)}\}\tilde{G}^{(Q_{2})}_{1}\tilde{G}^{(Q_{2})}_{2}\ket{\tilde{\phi}}=0,
\end{eqnarray}
where we used the fact that $\tilde{G}^{(Q_{2})}_{1}$ and $\tilde{G}^{(Q_{2})}_{2}$ commute which is a consequence of the
fact that $G_1^{\max}$ and $G_2^{\max}$ commute 
[cf. Eqs. (\ref{G1max}) and (\ref{G2max})].

Let us now assume that 
$\{\tilde{X}^{(i)},\tilde{Z}^{(i)}\}=0$ for some $i\geq 2$. Our aim is to 
prove that $\{\tilde{X}^{(i+1)},\tilde{Z}^{(i+1)}\}=0$. 
As mentioned before the regular generators $G_{1}^{\max},\dots,G_{k_{\min}(N)}^{\max}$ stabilizing $V_N^{\max}$ give rise to vectors (\ref{vectorsW}) for every $i\in\{1,\dots,N-1\}$. This by the very definition of these vectors implies that for every pair of adjacent qubits there always exists a pair of generators $G^{\max}_{p}$ and $G^{\max}_{q}$ $(p\neq q)$ that anticommute on exactly these two qubits. Consequently, for subsystems $i$ and $i+1$ the corresponding $\tilde{G}_{j}$ and $\tilde{G}_{q}$ are of the following form
\begin{eqnarray}
\tilde{G}_{p}&=&P_{a}^{(i)}P_{b}^{(i+1)}G_{p}^{(Q_{i,i+1})},\nonumber\\
\tilde{G}_{q}&=&P_{a+1}^{(i)}P_{b+1}^{(i+1)}G_{q}^{(Q_{i,i+1})}
\end{eqnarray}
for some $a,b\in\{0,1\}$, where $Q_{i,i+1}=\{1,\ldots,N\}\setminus \{i,i+1\}$.
We thus have the following chain of equalities
\begin{eqnarray}
\{\tilde{P}_{b}^{(i+1)},\tilde{P}_{b+1}^{(i+1)}\}\ket{\tilde{\phi}}&=&(\tilde{P}_{b}^{(i+1)}\tilde{P}_{b+1}^{(i+1)}\tilde{G}_{q}\tilde{G}_{j}\nonumber\\
&&\hspace{0.5cm}+\tilde{P}_{b+1}^{(i+1)}\tilde{P}_{b}^{(i+1)}\tilde{G}_{j}\tilde{G}_{q})\ket{\tilde{\phi}}\nonumber\\
&=&(\tilde{P}_{a+1}^{(i)}\tilde{P}_{a}^{(i)}\tilde{G}^{(Q_{i,i+1})}_{q}\tilde{G}^{(Q_{i,i+1})}_{j}\nonumber\\
&&\hspace{0.5cm}+\tilde{P}_{a}^{(i)}\tilde{P}_{a+1}^{(i)}\tilde{G}^{(Q_{i,i+1})}_{j}\tilde{G}^{(Q_{i,i+1})}_{q})\ket{\tilde{\phi}}\nonumber\\
&=&\{\tilde{P}_{a}^{(i)},\tilde{P}_{a+1}^{(i)}\}\tilde{G}^{(Q_{i,i+1})}_{j}\tilde{G}^{(Q_{i,i+1})}_{q}\ket{\tilde{\phi}}\nonumber\\
&=&0.
\end{eqnarray}

Now, we can use Lemma \ref{lem_selftest}, which tells us that 
for each site $i$ there exists a unitary matrix $U_i$ such that 
\begin{equation}
    U_i \tilde{X}^{(i)}U_i^{\dagger}=X_i\otimes \mathbbm{1}_{P''_i},\qquad
    U_i \tilde{Z}^{(i)}U_i^{\dagger}=Z_i\otimes \mathbbm{1}_{P''_i},
\end{equation}
where $X_i$ and $Z_i$ are the Pauli matrices acting on the $i$th site and $\mathbbm{1}_{P''_i}$ is an identity acting on $\mathcal{H}_{P''_i}$. Accordingly, 
we have
\begin{equation}
    U_P \tilde{G}_i U_P^{\dagger}= G_i^{\max}\otimes \mathbbm{1}_{P''}
\end{equation}
for $U_P=U_1\otimes\ldots\otimes U_N$. Denoting then 
$\ket{\psi}=(U_P\otimes\mathbbm{1}_E)\ket{\tilde{\phi}}$, 
the conditions (\ref{conds}) rewrite as
\begin{equation}
  (  G_i^{\max}\otimes\mathbbm{1}_E)\ket{\psi}=\ket{\psi}
\end{equation}
for $i=1,\ldots,k_{\min}(N)$. Taking the Schmidt decomposition
of $\ket{\psi}$ with respect to the subsystems $P$ and $E$, one finds that the most general form of $\ket{\psi}$ compatible with 
these conditions is the following one
\begin{equation}
    \ket{\psi}=\sum_i\alpha_i\ket{\psi_i}_{P'}\otimes \ket{\xi_i}_{P''E},
\end{equation}
with $\alpha_i$ being some nonnegative coefficients whose squares sum up to one. Thus, 
\begin{equation}
    (U_P\otimes\mathbbm{1}_E)\ket{\tilde{\phi}}=\sum_i\alpha_i\ket{\psi_i}_{P'}\otimes \ket{\xi_i}_{P''E},
\end{equation}
which completes the proof.
\end{proof}

\section{Exploring the boundaries of the sets of quantum correlations}
\label{boundaries}

A behaviour maximally violating a Bell inequality belongs to the boundary of the set of quantum correlations. Thus, Bell inequalities for whose maximal quantum violation is known can serve as a tool for characterizing the boundaries of the sets of quantum correlations. Here we use our inequalities to identify 
faces of dimension $2^{N-k_{\min}(N)}-1$ in the boundaries of the 
sets $\mathcal{Q}_{N}$. We reach this goal by constructing, for any $N$, a set of $2^{N-k_{\min}(N)}$ behaviours $\mathcal{P}_{i}$ achieving the maximal quantum value of (\ref{NierMax}) and prove them to be affinely independent. These in turn are obtained from a certain set of $2^{N-k_{\min}(N)}$ multipartite qubit states
stabilized by $G_i^{\max}$. To simplify further calculation we will not explicitly state the dependence of $k_{\min}$ on $N$, but keep in mind that this is still a function of $N$.

We establish this result in a few steps. The first one is to consider an additional set of $N-k_{\min}$ operators from the group $\mathbb{G}_N$ of 
the following form
\begin{equation}\label{eq_def_hij}
   H_{i,j}^{\max}=\mathbbm{1}_2^{\otimes \frac{1}{2}(i^2+i+2)}\otimes\mathbbm{1}_2^{\otimes (j-1)}
    \otimes P_{i+1}^{\otimes 2}\otimes \mathbbm{1}_2^{\otimes [N-\frac{1}{2}(i^2+i)-j-2]},
\end{equation}
where $i=1,\ldots,i_{\max}$ with $i_{\max}$ being the largest $i$ fulfilling $i(i+1)/2<N-k_{\min}$, whereas $j=1,\ldots,N-k_{\min}-i_{\max}$. As shown in Appendix \ref{App:Ind} these operators are independent. Moreover, they mutually commute because the only two nontrivial matrices of $H_{i,j}^{\max}$, which are $P_{i+1}$, belong to the block $C_{i+2}$; so either two different $H_{i,j}^{\max}$ have their nontrivial matrices in different blocks, or their nontrivial matrices are the same. It is also not difficult to see that one cannot use $H_{i,j}^{\max}$ to generate $-\mathbbm{1}^{\otimes N}_2$. In fact, they give rise to a nontrivial stabilizer. 

Let us now consider the generators $G_{i}^{\max}$ together with $H_{i,j}^{\max}$.
Appendix \ref{App:Ind} establishes that they form a set of $N$ independent operators. It is also not difficult to see that each $G_{i}^{\max}$ commutes with each $H_{i,j}^{\max}$. This is because the only two nontrivial matrices $P_{i+1}$ of $H^{\max}_{i,j}$ belong to the block $C_{i+2}$. At the same time, the only generator of $\mathbb{S}_{N}$ that has matrices $P_{i}$ in this block is $G_{i+2}^{\max}$ and it has $P_{i}$ on every qubit from that block. So $H^{\max}_{i,j}$ and $G_{i+2}^{\max}$ anticommute on two qubits, hence they commute. The remaining generators from $\mathbb{S}_{N}$ also commute with $H^{\max}_{i,j}$, as either all of their matrices in $C_{i+2}$ are simply $\mathbbm{1}_2$, or one or two equal $P_{i+1}$, whereas the rest are $\mathbbm{1}_2$.
To prove that $G_{i}^{\max}$ and $H_{i,j}^{\max}$ together do not generate $-\mathbbm{1}^{\otimes N}_2$ we notice, as before, that one cannot generate the product $XZXZ$ at any site by any product of four different generators of $\overline{\mathbb{S}}_N$. Namely, there is only one generator that has a matrix $P_{i}$ on the block $C_{i}$, which is $G_{i}^{\max}$.
These two facts ensure that $G_{i}^{\max}$ and $H_{i,j}^{\max}$ generate a stabilizer of a one dimensional subspace. 

We can now use this set of $N$ independent and mutually commuting operators to 
construct the aforementioned multiqubit pure states $\ket{\psi_i}$. To this aim, 
we consider a set of $2^{N-k_{\min}}$ subgroups of $\mathbb{G}_N$ generated by $G_{i}^{\max}$ and $H_{i,j}^{\max}$, however, the latter we take with different signs $\pm 1$; in other words we consider the following sets
\begin{equation}\label{eq_stab_gh}
    \mathbb{S}_N^i=\left\langle \{G_{i}^{\max}\},\{\pm H_{j,k}^{\max}\}\right\rangle,
\end{equation}
where $i=1,\ldots,2^{N-k_{\min}(N)}$. A change of sign in front of the $H_{j,k}^{\max}$ operators has no effect on the commutation of the operators and so to show that every $\mathbb{S}_N^i$ stabilises a nontrivial subspace it is enough to show that $-\mathbbm{1}_2^{\otimes N}\notin \mathbb{S}_N^i$. To prove this let us assume that $-\mathbbm{1}_2^{\otimes N}\in \mathbb{S}_N^i$. All operators from $\mathbb{S}_N^i$ clearly square to $\mathbbm{1}_2^{\otimes N}$ and so $-\mathbbm{1}_2^{\otimes N}$ has to be generated as a product of non repeating generators of $\mathbb{S}_N^i$. This implies that some product of non repeating generators of $\mathbb{S}_N^0=\left\langle \{G_{i}^{\max}\},\{H_{i,j}^{\max}\}\right\rangle$ equals $\pm \mathbbm{1}_2^{\otimes N}$. However, it was already proven that $\mathbb{S}_N^0$ stabilises a nontrivial subspace, so any product of generators cannot be equal to $-\mathbb{1}$ and the subspace stabilised by $\mathbb{S}_{N}^{0}$ is also one-dimensional, so the product of non repeating generators cannot be equal to $\mathbb{1}$. Hence every $\mathbb{S}_N^i$ stabilises a nontrivial subspace.

Similarly we can show that every $\mathbb{S}_N^i$ stabilises a one dimensional subspace. This requires the generators of $\mathbb{S}_N^i$ to be mutually independent. Let us assume that this is not a case, and so there exists some product of generators of $\mathbb{S}_N^i$ that equals $\mathbbm{1}_2^{\otimes N}$. By the same argument as in the previous paragraph we can show that this contradicts the fact that $\mathbb{S}_N^0$ stabilises a nontrivial, one dimensional subspace.
Now, each $\mathbb{S}_N^i$ stabilizes a
multiqubit state $\ket{\psi_i}$ that clearly belongs to $V_N^{\max}$. Moreover, since states $\ket{\psi}_{i}$ are stabilised by (\ref{eq_stab_gh}) it follows that
\begin{equation}
\langle\psi_i| G_{k}^{\max}|\psi_i\rangle   =1
\end{equation}
for any $i=1,\ldots, 2^{N-k_{\min}(N)}$ and $k=1,\ldots,k_{\min}$, whereas
\begin{equation}\label{eq_orthonormal}
\langle\psi_i| H_{k,l}^{\max}|\psi_i\rangle   =\pm 1.
\end{equation}
We can define behaviours $\mathcal{P}_{i}$ corresponding to the states $\ket{\psi_{i}}$ and a choice of observables (\ref{ChoiceN3}) as in (\ref{correlations}). Furthermore, we denote vector elements of the behaviour $\mathcal{P}_{i}$ by:
\begin{equation}
\left\langle \prod_{j\in J}A_{x_{j}}^{(j)}\right\rangle_{\mathcal{P}_{i}}\equiv \bra{\psi_{i}}\prod_{j\in J}^{N}A_{x_{j}}^{(j)}) \ket{\psi_{i}},
\end{equation}
where $x_{j}\in\{0,1\}$, $J\subset \{1,\dots,N\}$ and we chose the observables as in (\ref{ChoiceN3}). To show that all $\mathcal{P}_{i}$ are affinely independent let us first relabel $H_{j,k}^{\max}\to H_{i}$ where $i\in\{1,\dots, N-k_{\min}\}$. The exact correspondence between them is not given, because it will not be necessary to know the exact form of operators $H_{i}$. Next, we denote by $\tilde{H}_{i}$ the quantum operators obtained from generators $H_{i}$ by the substitution (\ref{assign1}), (\ref{assign22}): $H_{i}\rightarrow \tilde{H}_{i}$. Finally, in order to show that all $\mathcal{P}_{i}$ are affinely independent, it is sufficient to only consider the following expected values
\begin{equation}\label{eq_corr}
\left\langle\prod_{i\in J} \tilde{H}_{i} \right\rangle,
\end{equation}
for all nonempty $J\subset \{1,\dots,2^{N-k_{\min}}-1\}$ and one of the expectation values $\langle \tilde{G}_{i} \rangle$. The operator $\tilde{G}_{i}$ is an observable created from the operator $G_{i}^{\max}$ by the substitution (\ref{assign1}), (\ref{assign22}). It is irrelevant which $\tilde{G}_{i}$ we chose, because for all $\mathcal{P}_{j}$ we have $\langle \tilde{G}_{i} \rangle_{\mathcal{P}_{j}}=1$. Let us also mention here that (\ref{eq_corr}) are well defined expectation values, because on every qubit we either have $\mathbb{1}$, $A_{0}$ or $A_{1}$. This is true since from (\ref{eq_def_hij}) we see that if two operators $H_{i,j}^{\max}$ have different index $i$ then on no qubit do they both have a nontrivial matrix, and if they have the same $i$ then they have the same nontrivial matrices $P_{i+1}$. Hence the product of $H_{i}$ operators will never result in an operator that has a matrix proportional to $Y$ on some qubit.

It is also important to note that since we consider only behaviours $\mathcal{P}_{j}$ obtained from the states $\ket{\psi_{j}}$ stabilised by (\ref{eq_stab_gh}), we can write for every $\mathcal{P}_{j}$,
\begin{equation}\label{culo}
\left\langle\prod_{i\in M} \tilde{H}_{i} \right\rangle_{\mathcal{P}_{j}}=\prod_{i\in M}\left\langle\tilde{H}_{i} \right\rangle_{\mathcal{P}_{j}}.
\end{equation}

Now, let us move on to proving that all $\mathcal{P}_j$ are linearly independent. To this aim, let us assume that for some $q\in\{1,\dots, 2^{N-k_{\min}}-1\}$ we have
\begin{equation}\label{eq_beh}
\mathcal{P}_{q}=\sum_{i\in M} a_{i}\mathcal{P}_{i},
\end{equation}
for some subset $M\subset \{1,\dots, 2^{N-k_{\min}}-1\}$ and $a_{i}\in\mathbb{R}$. Every expected value (\ref{eq_corr}) gives us one condition on the values of $a_{i}$. We can add all these conditions, which gives us 
\begin{eqnarray}
\sum_{i\in M}a_{i}(1+\langle \tilde{H}_{1}\rangle_{\mathcal{P}_{i}}+\dots + \langle \tilde{H}_{1}\dots \tilde{H}_{N-k_{\min}}\rangle_{\mathcal{P}_{i}})\nonumber\\=
1+\langle \tilde{H}_{1}\rangle_{\mathcal{P}_{q}}+\dots + \langle \tilde{H}_{1}\dots \tilde{H}_{N-k_{\min}}\rangle_{\mathcal{P}_{q}},\nonumber\\
\end{eqnarray}
which, using Eq. (\ref{culo}), can be rewritten as 
\begin{equation}\label{eq_iinn}
\sum_{i\in M}a_{i}\prod_{j=1}^{N-k_{\min}}(1+\langle \tilde{H}_{j}\rangle_{\mathcal{P}_{i}})=\prod_{j=1}^{N-k_{\min}}(1+\langle \tilde{H}_{j}\rangle_{\mathcal{P}_{q}}).
\end{equation}
Clearly, the only nonzero term comes from the behaviour $\mathcal{P}_{i}$ for which all $\langle \tilde{H}_{j}\rangle=1$ for all $j$. Then if $i\neq q$ we have $a_{i}=0$ and if $i=q$ then $a_{i}=1$. The fact that the only nonzero term in (\ref{eq_iinn}) is the one for which all $\langle \tilde{H}_{i}\rangle=1$ can be exploited to calculate the rest of coefficients $a_{i}$. Equation (\ref{eq_iinn}) is a result of adding all conditions from
(\ref{eq_beh}), however we can also subtract the conditions. To calculate the value of $a_{m}$ we add a condition from an expectation value $\langle \prod_{i\in M'} \tilde{H}_{i}\rangle$ if $\langle \prod_{i\in M'} \tilde{H}_{i}\rangle_{\mathcal{P}_{m}}=1$ and subtract it if $\langle \prod_{i\in M'} \tilde{H}_{i}\rangle_{\mathcal{P}_{m}}=-1$. As a result we get the following equation
\begin{eqnarray}
\sum_{i\in M}a_{i}\prod_{j=1}^{N-k_{\min}}(1+&&\langle \tilde{H}_{j}\rangle_{\mathcal{P}_{m}}\langle \tilde{H}_{j}\rangle_{\mathcal{P}_{i}})=\\ \nonumber
&&=\prod_{j=1}^{N-k_{\min}}(1+\langle \tilde{H}_{j}\rangle_{\mathcal{P}_{m}}\langle \tilde{H}_{j}\rangle_{\mathcal{P}_{q}}).  
\end{eqnarray}
From this it follows that $a_{m}=0$ for $m\neq q$ and $a_{m}=1$ for $m=q$, and since we can construct this equation for all $m\in M$, all $\mathcal{P}_{i}$ are affinely independent. This in turn means that the face of the quantum correlation set corresponding to the Bell inequality (\ref{NierMax}) is $(2^{N-k_{\min}}-1)$-dimensional.

\section{Conclusions}
\label{Concl}
Here we have shown that genuinely entangled stabilizer subspaces of maximal dimension can be self-tested from maximal violation of certain multipartite Bell inequalities. We have established this result in a few steps. First, we have introduced a framework allowing to decide whether stabilizer subspaces in  multiqubit Hilbert space are genuinely entangled. We have then 
determined what is the maximal dimension of genuinely entangled subspaces
achievable within the stabilizer formalism, and, using the above framework, we have provided a construction of such a maximally-dimensional 
stabilizer subspace. Third, we constructed Bell inequalities maximally violated by any pure state from these subspaces and showed that 
maximal violation of these inequalities can be used to self-test them. 
Finally, we have used our Bell inequalities to identity 
"flat structures" in the boundaries of the sets of quantum correlations in the considered Bell scenarios.

Our work opens many possibilities for further research. 
First, we provide self-testing strategies only for particular stabilizer subspaces and it is interesting to explore whether our approach can be used to self-test any genuinely entangled subspace within the stabilizer formalism. 
Then, a much more challenging problem would be of course to design self-testing strategies for entangled subspaces that do not fall into the stabilizer formalism; recall, however, that even a simpler problem of whether any genuinely entangled multipartite state can be self-tested remains open. On the other hand, it seems that our formalism to describe entanglement within the stabilizer formalism can be generalized to the case of multi-qudit Hilbert spaces of local dimension being a prime number or a power of prime (see, e.g., Ref. \cite{Hein}). It thus would be interesting to explore whether such a formalism is of use in studying entanglement properties of stabilizer subspaces, in particular those corresponding to quantum error codes. At the same time, by combining it with the construction of Bell inequalities maximally violated by the two-qudit maximally entangled states of Ref. \cite{Quantum} one could perhaps design self-testing methods for genuinely entangled subspaces in multiqudit Hilbert spaces.

\appendix

\section{A proof of Fact \ref{fakt2}}
\label{AppA}

Here we provide a proof of Fact \ref{fakt2}, i.e., that 
the operators defined in Eqs. (\ref{eq_stab}) generate a subspace of maximal dimension $2^{N-k_{\min}(N)}$. To this end, we show that they mutually commute and
that $-\mathbbm{1}^{\otimes N}_2\notin \langle G_i^{\max}\rangle$.
To simplify the notation in what follows we drop the subscript "max"
from the generators $G_{i}^{\max}$.

To prove the first condition let us begin by dividing our $N$ qubits
$\{1,\ldots,N\}$ into $k=k_{\min}(N)$ disjoint and nonempty blocks defined as 
\begin{equation}\label{eq_cluster}
C_{i}=\begin{cases}
\{i\},& i=1,2,\\[1ex]
\{2+\frac{(i-1)(i-2)}{2},\dots,1+\frac{i(i-1)}{2}\},& i=3,\dots,k-1,\\[1ex]
\{2+\frac{(k-1)(k-2)}{2},\dots,N\},& i=k,
\end{cases}
\end{equation}
As the number of qubits in the block $C_{i}$ equals
\begin{equation}\label{eq_cluster_size}
|C_{i}|=\begin{cases}
1,& i=1,\\[1ex]
i-1,& i=2,\dots,k-1,\\[1ex]
N-\frac{(k-1)(k-2)}{2}-1,& i=k.
\end{cases}
\end{equation}
one verifies that
\begin{equation}
 \sum_{i=1}^{k}|C_i|=N.   
\end{equation}

Now, let us examine the form of $G_{i}^{(C_{l})}$ for different $l\in\{1,\dots,k-2\}$. The first term in Eq. (\ref{eq_stab}) is $\mathbbm{1}^{\otimes(i-1)(i-2)/2}_2$ and therefore we have for $i\geq 3$:
\begin{equation}
G_{i}^{(C_{l})}=\mathbbm{1}_2^{\otimes(l-1+\delta_{l,1})}\quad \text{for }\quad l=1,\dots,i-2,
\end{equation}
and
\begin{equation}\label{form1}
G_{i}^{(C_{i-1})}=\mathbbm{1}_2^{\otimes(i-3)}\otimes P_{i}.
\end{equation}
The matrix $P_{i}$ on the last qubit of $G_{i}^{(C_{i-1})}$ comes from the second term in Eq. (\ref{eq_stab}), i.e., $P_{i}^{\otimes(i+1)}.$
The first matrix of this term is in the block $C_{i-1}$ and from Eq. (\ref{eq_cluster_size}) we infer that:
\begin{align*}
\begin{split}
G_{i}^{(C_{i})}&=P_{i}^{\otimes(i-1)},\\
G_{i}^{(C_{i+1})}&=P_{i}\otimes G_{i}^{(C_{i+1}\setminus\{2+(i-1)(i-2)/2\})}.
\end{split}
\end{align*}
As we can see, this gives us only the first matrix of $G_{i}^{(C_{i+1})}$. To see the whole $G_{i}^{(C_{i+1})}$ we need to look at the third term in Eq. (\ref{eq_stab}), i.e., 
\begin{equation}\label{eq_term_3}
\bigotimes_{q=i+2}^{q_{\max}(i,N)}\big(\mathbbm{1}_2^{\otimes(q-3)}\otimes P_{q+1}^{\otimes 2}\big).
\end{equation}

For $q=i+2$ we have $i-1$ of $\mathbbm{1}_2$ in a row. From Eq. (\ref{eq_cluster_size}) we know that $|C_{i+1}|=i$ for 
$i\in\{1,\dots,k-2\}$ so we have
\begin{equation*}
G_{i}^{(C_{i+1})}=P_{i}\otimes\mathbbm{1}^{\otimes(i-1)}_2. 
\end{equation*}
Next, we need to find $G_{i}^{(C_{l})}$ for $l\in\{i+2,\dots,k-1\}$. In (\ref{eq_term_3}) the nontrivial matrices occur in pairs and the qubits of the $(l-i)$th pair [so for $q=l+1$ in Eq. (\ref{eq_term_3})] are as follows
\begin{eqnarray*}
\hspace{-3cm}&&\left\{\frac{1}{2}(i-1)(i-2)+i+\frac{1}{2}(i+1+l)(l-1),\right.\nonumber\\
&&\left.\frac{1}{2}(i-1)(i-2)+i+1+\frac{1}{2}(i+1+l)(l-1)\right\}.
\end{eqnarray*}
Comparing this to qubits of $C_{l}$ from (\ref{eq_cluster}) we can conclude that the $(l-i)$th pair is present in $C_{l}$ in the following way:
\begin{equation*}
G_{i}^{(C_{l})}=\mathbbm{1}_2^{\otimes(l-i-2)}\otimes P_{l+1}^{\otimes2}\otimes\mathbbm{1}^{\otimes(i-1)}_2,
\end{equation*}
As for the last block $C_{k}$ we have two cases. First, if
$q_{\max}(i,N)\cdot[q_{\max}(i,N)+1]/2> N-1$, then
\begin{equation}\label{eq_case_1}
G_{i}^{(C_{k})}=\mathbbm{1}_2^{\otimes(k-i-2)}\otimes P_{k-1}^{\otimes2}\otimes\mathbbm{1}_2^{\otimes(i-\frac{1}{2}k(k-1)+N-2)}.
\end{equation}
On the other hand, if $q_{\max}(i,N)\cdot[q_{\max}(i,N)+1]/2\leq N-1$, then
\begin{equation}\label{eq_case_2}
G_{i}^{(C_{k})}=\mathbbm{1}_2^{\otimes\left[k-\frac{1}{2}k(k-1)+N-2
\right]}.
\end{equation}

\begin{widetext}
Summarizing, 

\begin{equation}\label{eq_stab_i}
G_{i}^{(C_{l})} = \begin{cases}
\mathbbm{1}_2^{\otimes(l-1+\delta_{l,1})}, & l=1,\dots,i-2,\\
\mathbbm{1}_2^{\otimes(i-3)}\otimes P_{i}, & l=i-1,\\
P_{i}^{\otimes(i-1)}, &l=i,\\
\mathbbm{1}_2^{\otimes(l-i-2+\delta_{l-i,1})}\otimes P_{l+1}^{\otimes(2-\delta_{l-i,1})}\otimes\mathbbm{1}_2^{\otimes(i-1)}, &l=i+1,\dots,k-1,
\end{cases}
\end{equation}
and for $l=k$:
\begin{equation}\label{eq_stab_i_n}
G_{i}^{(C_{k})}=\begin{cases}
\mathbbm{1}_2^{\otimes(k-i-2)}\otimes P_{k-1}^{\otimes2}\otimes\mathbbm{1}_2^{\otimes(i-\frac{1}{2}k(k-1)+N-2)}, & q_{\max}(i,N)\cdot[q_{\max}(i,N)+1]/2> N-1,\\
\mathbbm{1}_2^{\otimes\big(k-\frac{1}{2}k(k-1)+N-2\big)}, & q_{\max}(i,N)\cdot[q_{\max}(i,N)+1]/2\leq N-1,
\end{cases}
\end{equation}
By the same argument we have
\begin{equation}\label{eq_stab_n-1}
G_{k-1}^{(C_{l})} = \begin{cases}
\mathbbm{1}_2^{\otimes(l-1+\delta_{l,1})}, &l=1,\dots,k-3,\\
\mathbbm{1}_2^{\otimes(k-4)}\otimes P_{k-1}, &l=k-2,\\
P_{k-1}^{\otimes(k-2)}, &l=k-1,\\
P_{k-1}\otimes\mathbbm{1}_2^{\otimes\left[k-\frac{1}{2}k(k-1)+N-3\right]}, &l=k,
\end{cases}
\end{equation}
and
\begin{equation}\label{eq_stab_n}
G_{k}^{(C_{l})} = \begin{cases}
\mathbbm{1}_2^{\otimes(l-1+\delta_{l,1})}, &l=1,\dots,k-2,\\
\mathbbm{1}_2^{\otimes(k-3)}\otimes P_{k}, &l=k-1,\\
P_{k}^{\otimes\left[k-\frac{1}{2}k(k-1)+N-2\right]}, &l=k.\\ 
\end{cases}
\end{equation}
\end{widetext}
Now, we can use this description of the generators $G_{i}^{\max}$ to show they mutually commute. Using Eqs. (\ref{eq_stab_i}), (\ref{eq_stab_i_n}), (\ref{eq_stab_n-1}) and (\ref{eq_stab_n-1}) one can check that the following holds true
\begin{equation}
\left\{G_{i}^{(C_{l})},G_{i+1}^{(C_{l})}\right\}=0\qquad (l=i,i+1),
\end{equation}
and
\begin{eqnarray}
&&\hspace{-1cm}\left[G_{i}^{(C_{l})},G_{i+1}^{(C_{l})}\right]=0\nonumber\\
&&\hspace{1cm}(l=1,\dots i-2,i+1,\dots,k),
\end{eqnarray}
for all $i=1,\dots,k-1$. Similarly we have
\begin{equation}
\left[G_{j}^{(C_{l})},G_{i}^{(C_{l})}\right]=0\qquad
(l=1,\dots,k),
\end{equation}
for all $i=1,\dots,k$ and all $j=1,\dots,i-2$. From this we can infer that
\begin{equation}
\left[G_{i}^{\max},G_{j}^{\max}\right]=0  
\end{equation}
is true for all $i,j=1,\dots,k$.

Let us now show that it is not possible to generate $-\mathbbm{1}^{\otimes N}_2$ from $G^{\max}_i$. 
Due to the fact that $G_i^{\max}$ are composed of only $X$ and $Z$,
$-\mathbbm{1}_2^{\otimes N}$ can only be generated from the product $XZXZ$. 
Due to the fact that $G_i^{\max}$ mutually commute and square to identity, 
such a product can be achieved only from four distinct generators. 
To see whether the product $XZXZ$ can be constructed, let us analyse $(G_{i}^{\max})^{(C_{l})}$ for every $i=1,\dots,k$ and some $l$. From Eqs. (\ref{eq_stab_i}), (\ref{eq_stab_i_n}), (\ref{eq_stab_n-1}) and (\ref{eq_stab_n}) we can see that the only generator that has matrices $P_{i}$ on the block $C_{l}$ is $G_{l}^{\max}$ and the rest have either $P_{i+1}$ or $\mathbbm{1}_2$. Hence, 
the product $P_{i}P_{i+1}P_{i}P_{i+1}$ cannot be constructed with four 
distinct generators $G_i^{\max}$.

Next, we need to prove that our stabilizer $\mathbb{S}=\langle G_{1},\dots,G_{k}\rangle$ stabilises a GME subspace. Let us consider two generators $G_{i}$ and $G_{j}$ and, without any loss of generality, let us assume that $i>j$. As shown before they commute, but now we are interested in anticommutations on the single qubit. Let us once again consider blocks of qubits indexed by $l$. From (\ref{eq_stab_i}), (\ref{eq_stab_i_n}) (\ref{eq_stab_n-1}) and (\ref{eq_stab_n}) we can conclude that the generator $G_{i}$ anticommutes on single qubit with the generator $G_{j}$ in blocks $C_{i-1}$ and $C_{i}$. For $j\in\{1,\dots,i-2\}$ they anticommute on two consecutive qubits from $C_{i}$. For $j=i-1$ they anticommute on the last qubit of $C_{i-1}$ and the first of $C_{i}$ so those two qubits are also consecutive. This is true for any pair of generators, so we have $N-1$ vectors from $W_{N,d}$ of the form $w_{p}=\hat{e}_{p}+\hat{e}_{p+1}$ where $p\in\{1,\dots,N-1\}$. The last step is to show that these vectors are linearly independent. We will attempt that by expressing $i$ and $j$ in terms of the index $p$. We still assume that $i>j$ which means that $(p+1)$'th qubit has to be a part of $i$'th block. This gives us a condition on $p$
\begin{equation*}
p+1>1+\frac{i}{2}(i-1) \quad \Leftrightarrow \quad p>\frac{i}{2}(i-1),
\end{equation*}
where the right side is just a number of stabilizing operators in blocks with index smaller that $i$. For $(p+1)$'th qubit to be a part of $i$'th block, we have to take the smallest positive integer $i$ that fulfils this inequality, and hence we get
\begin{equation*}
i=\left\lceil\frac{1+\sqrt{1+8p}}{2}\right\rceil.
\end{equation*}
In the $i$'th block, the generator $G_{j}$ has it's second $P_{i+1}$ matrix on the qubit of an index:
\begin{equation*}
1+\frac{i}{2}(i-1)+i-j,
\end{equation*}
but it is also equal to $p+1$ so we can write
\begin{equation*}
j=\frac{i}{2}(i+1)-p.
\end{equation*}
So to summarise, for every vector $w_{p}=\hat{e}_{p}+\hat{e}_{p+1}$ where $p\in\{1,\dots,N-1\}$ we can find $i$ and $j$ such that $G_{i}$ and $G_{j}$ anticommute only on qubits $p$ and $p+1$ so $v_{i,j}=w_{p}$. Since vectors $w_{p}$ form a basis in $W_{N}$ we have $|K(\mathbb{S})|=N-1$.

We finally notice that $G_i^{\max}$ are certainly independent because
we already know that they stabilize a GME subspace and if they were not independent they would necessarily stabilise a subspace of  
dimension violating the upper bound in Eq. (\ref{BoundDim}).

\section{Bell inequalities for GME stabilizer subspaces from Sec. \ref{SecIIIB}}
\label{AppB}

Here we construct Bell inequalities maximally violated by 
the subspaces $\mathcal{V}_N$ introduced in Sec. \ref{SecIIIB}.

Consider the stabilizing operators $\mathcal{G}_i$ given in Eq. 
(\ref{wzory}). To all Pauli matrices $X$ and $Z$ appearing in them 
we assign, as before, the observables $A_j^{(i)}$ or their combinations. Precisely, for the first observer we make the following assignments
\begin{equation}
X\to A_{0}^{(1)}+A_1^{(1)},\qquad
Z\to A_{0}^{(1)}-A_1^{(1)},
\end{equation}
whereas for the remaining observers $P_i$ $(i=2,\ldots,N)$:
\begin{equation}
X\to A_{0}^{(i)}, \qquad Z\to A_{1}^{(i)}
\end{equation}

Then, using these assignments, we associate an expectation value to each generator $\widetilde{G}_i$. Precisely, to the first, the second, and the last ones we associate
\begin{eqnarray}
\mathcal{G}_1&\to& \langle (A_{0}^{(1)}+A_1^{(1)})A_0^{(2)}\ldots A_0^{(N)}\rangle,\nonumber\\
\mathcal{G}_2&\to& \langle (A_{0}^{(1)}-A_1^{(1)})A_1^{(2)}A_0^{(3)} A_0^{(N)}\rangle,\nonumber\\
\mathcal{G}_{N/2+1}&\to& \langle (A_{0}^{(1)}+A_1^{(1)})A_0^{(N-2)}A_1^{(N-1)} A_1^{(N)}\rangle,\nonumber\\
\end{eqnarray}
while to the remaining ones 
\begin{equation}
    \mathcal{G}_i\to \langle A_0^{(2i-4)}A_1^{(2i-3)}A_1^{(2i-2)}A_0^{(2i-1)}\rangle
\end{equation}
for $i=3,\ldots,N/2$. Adding the obtained expectation values and multiplying the first one by two, we obtain the following Bell inequality
\begin{widetext}
\begin{eqnarray}
  \mathcal{I}_{N}&:=&  \langle (A_{0}^{(1)}+A_1^{(1)})A_0^{(2)}\ldots A_0^{(N)}\rangle+
    2\langle (A_{0}^{(1)}-A_1^{(1)})A_1^{(2)}A_0^{(3)} A_0^{(N)}\rangle+\langle (A_{0}^{(1)}+A_1^{(1)})A_0^{(N-2)}A_1^{(N-1)} A_1^{(N)}\rangle\nonumber\\
    &&+\sum_{i=3}^{N/2}\langle A_0^{(2i-4)}A_1^{(2i-3)}A_1^{(2i-2)}A_0^{(2i-1)}\rangle\leq \frac{N}{2}+2,
\end{eqnarray}
where the maximal classical value $\widetilde{\beta}_C^{N}=N/2+2$ of this Bell expression was found by optimizing $\mathcal{I}_{N}$ over all local deterministic correlations. 

Importantly, the CHSH-like construction of this inequality makes it also possible to directly determine the maximal quantum value of $I_{\mathrm{ex}}$. Namely, we have the following observation.
\begin{fakt}
The maximal quantum value of $\mathcal{I}_{N}$ is
\begin{equation}\label{betaQ}
 \widetilde{ \beta}_Q^{N}:=  \max_Q \mathcal{I}_{N}=\frac{N}{2}+2(2\sqrt{2}-1).
\end{equation}
\end{fakt}
\begin{proof}Let $B_N$ be a Bell operator obtained from the Bell expression $\mathcal{I}_{N}$ for an arbitrary choice of the local dichotomic observables $A_{0/1}^{(i)}$. Then, 
one checks by hand that following sum-of-squares decomposition holds true:
\begin{eqnarray}
\left[\frac{N}{2}+2(2\sqrt{2}-1)\right]\mathbbm{1}-B_N
&=&\frac{1}{\sqrt{2}}\left(\mathbbm{1}-\frac{A_{0}^{(1)}+A_1^{(1)}}{\sqrt{2}}A_0^{(2)}\ldots A_0^{(N)}\right)^2+\sqrt{2}\left(\mathbbm{1}-\frac{A_{0}^{(1)}-A_1^{(1)}}{\sqrt{2}}A_1^{(2)}A_0^{(3)} A_0^{(N)}\right)^2\nonumber\\
&&+\frac{1}{\sqrt{2}}\left(\mathbbm{1}-\frac{A_{0}^{(1)}+A_1^{(1)}}{\sqrt{2}}A_0^{(N-2)}A_1^{(N-1)} A_1^{(N)}\right)^2\nonumber\\
&&+\frac{1}{2}\sum_{i=3}^{N/2}\left(\mathbbm{1}-A_0^{(2i-4)}A_1^{(2i-3)}A_1^{(2i-2)}A_0^{(2i-1)}\right)^2,\nonumber\\
\end{eqnarray}
meaning that  
$\max_Q \mathcal{I}_{N}\leq [N/2+2(2\sqrt{2}-1)]$. To prove this 
bound to be tight we can easily construct a quantum realization for which the value of $\mathcal{I}_{N}$ equals 
$N/2+2(2\sqrt{2}-1)$. In fact, by taking, as before,
\begin{equation}
A_0^{(1)}=\frac{X+Z}{\sqrt{2}},\qquad    
A_1^{(1)}=\frac{X-Z}{\sqrt{2}}
\end{equation}
and $A_{0/1}^{(i)}=X/Z$ for the remaining parties $i=2,\ldots,N$, we bring the Bell operator $B_N$ to a weighted sum of the stabilizing operators $\widetilde{G}_i$, i.e., $B=\sqrt{2}(\mathcal{G}_1+2\mathcal{G}_2+\mathcal{G}_{N/2+1})+ \mathcal{G}_{3}+\ldots+\mathcal{G}_{N/2}$, whose maximal eigenvalue is $N/2+2(2\sqrt{2}-1)$ and it is achieved by any state stabilized by $\mathcal{G}_i$.
\end{proof}
\end{widetext}

It is worth noting that $\widetilde{\beta}_C^{N}<\widetilde{\beta}_Q^{N}$ for any even $N\geq 4$, and thus any Bell inequality in this construction is nontrivial.

\section{Proof of independence}
\label{App:Ind}

Our aim here is to show that $G_i^{\max}$ together with 
$H_{i,j}^{\max}$ form a set of $N$ independent generators.
To do this let us first introduce a simple condition that 
enables verifying independent of elements of $\mathbb{G}_N$.

Let us consider a set of operators $G_i\in \mathbb{G}_N$
$(i=1,\ldots,k)$ that for simplicity we assume to be 
composed of $X$, $Z$ or $\mathbbm{1}_2$. We say that 
an operator $G_i$ has a unique matrix $P_j$ on some qubit
if the every other operator $G_i$ has either $P_{j+1}$
or $\mathbbm{1}_2$ on this qubit, where $P_{i}$ is defined in 
Eq. (\ref{Pi}). We can now state our criterion.
\begin{fakt}
If every $G_{i}$ $(i=1,\ldots,k)$ has a unique matrix at some position, then 
the operators $G_i$ are independent.
\end{fakt}
\begin{proof}
Let us say that generator $G_{i}$ has a unique matrix $P_{j}$ on the $k$'th qubit. Then by taking products of the remaining operators $G_j$ $(j\neq i)$ one cannot construct $G_i$; only operators that on the $k$'th qubit have either $P_{j+1}$ or $\mathbbm{1}_2$ can be created from them.
\end{proof}

This simple observation will be of a great use for our goal. 
Let us therefore show that all the operators $G_i^{\max}$ stabilizing 
the the maximally-dimensional GME subspaces $V_N^{\max}$ as well as
$H_{i,j}^{\max}$ have unique matrices. 

We begin with $G_{i}^{\max}$ and consider generators given by Eqs. (\ref{eq_stab_i}), (\ref{eq_stab_i_n}), (\ref{eq_stab_n-1}) and (\ref{eq_stab_n}). The unique matrices for a generator $G_{i}^{\max}$ are in the block $C_{i}$, because only this generator has matrices $P_{i}$ in $C_{i}$. The rest of generators on $C_{i}$ has either $\mathbb{1}$ or $P_{i+1}$. As an example we can look at the generators for $N=7$:
\begin{equation}\label{Gi7}
\begin{tabular}{*{15}{@{\hspace{.5mm}}c@{\hspace{.5mm}}}}
$\displaystyle G_{1}^{\max}$ & = & $\textcolor{red}{X}$ & $\otimes$ & $X$ & $\otimes$ & $Z$ & $\otimes$ & $Z$ & $\otimes$ & $\mathbbm{1}_2$ & $\otimes$ & $X$ & $\otimes$ & $X$,\\[1ex]
$\displaystyle G_{2}^{\max}$ & = & $Z$ & $\otimes$ & $\textcolor{red}{Z}$ & $\otimes$ & $Z$ & $\otimes$ & $\mathbbm{1}_2$ & $\otimes$ & $X$ & $\otimes$ & $X$ & $\otimes$ & $\mathbbm{1}_2$, \\[1ex]
$\displaystyle G_{3}^{\max}$ & = & $\mathbbm{1}_2$ & $\otimes$ & $X$ & $\otimes$ & $\textcolor{red}{X}$ & $\otimes$ & $\textcolor{red}{X}$ & $\otimes$ & $X$ & $\otimes$ & $\mathbbm{1}_2$ & $\otimes$ & $\mathbbm{1}_2$, \\[1ex]
$\displaystyle G_{4}^{\max}$ & = & $\mathbbm{1}_2$ & $\otimes$ & $\mathbbm{1}_2$ & $\otimes$ & $\mathbbm{1}_2$ & $\otimes$ & $Z$ & $\otimes$ & $\textcolor{red}{Z}$ & $\otimes$ & $\textcolor{red}{Z}$ & $\otimes$ & $\textcolor{red}{Z}$. 
\end{tabular}
\end{equation}
The matrices highlighted in red are unique for their specific generators. 

In a similar way we can show that the additional stabilizing operators
$H_{i,j}^{\max}$ are independent of $G_i^{\max}$.
For a clearer view we first consider the example for $N=7$:
\begin{equation}
\begin{tabular}{*{15}{@{\hspace{.5mm}}c@{\hspace{.5mm}}}}
$\displaystyle H^{\max}_{1,1}$ & = & $\mathbbm{1}_2$ & $\otimes$ & $\mathbbm{1}_2$ & $\otimes$ & $Z$ & $\otimes$ & $Z$ & $\otimes$ & $\mathbbm{1}_2$ & $\otimes$ & $\mathbbm{1}_2$ & $\otimes$ & $\mathbbm{1}_2$, \\[1ex]
$\displaystyle H_{2,1}^{\max}$ & = & $\mathbbm{1}_2$ & $\otimes$ & $\mathbbm{1}_2$ & $\otimes$ & $\mathbbm{1}_2$ & $\otimes$ & $\mathbbm{1}_2$ & $\otimes$ & $X$ & $\otimes$ & $X$ & $\otimes$ & $\mathbbm{1}_2$, \\[1ex]
$\displaystyle H_{2,2}^{\max}$ & = & $\mathbbm{1}_2$ & $\otimes$ & $\mathbbm{1}_2$ & $\otimes$ & $\mathbbm{1}_2$ & $\otimes$ & $\mathbbm{1}_2$ & $\otimes$ & $\mathbbm{1}_2$ & $\otimes$ & $X$ & $\otimes$ & $X$ .
\end{tabular}
\end{equation}
It is straightforward to see that for $N=7$ these new generators are independent of (\ref{Gi7}) and it is not more difficult to draw the same conclusion for any $N$. To this end, we can once again utilise the blocks $C_{l}$. From Eq. (\ref{eq_cluster_size}) we have that the first $i+1$ blocks together contain $(i^{2}+i+2)/2$ qubits, so on all of these blocks $H^{\max}_{i,j}$ has identities. The block $C_{i+2}$ contains $i+1$ qubits, so for all $j\in\{0,\dots,i\}$ only nontrivial matrices of $H^{\max}_{i,j}$ are in the $C_{i+2}$ and these matrices are equal to $P_{i+1}$, so $H^{\max}_{i,j}$ does not contain any unique matrices of stabilizing operators (\ref{eq_stab}). Hence, all $H_{i,j}^{\max}$ are independent of the generators of $\mathbb{S}_{N}^{\max}$. 

Let us finally show that $H^{\max}_{i,j}$ also form a set of independent operators. As already stated all nontrivial matrices of $H^{\max}_{i,j}$ are in $C_{i+2}$, so $H^{\max}_{i,j}$ with different $i$ are clearly independent. We are left with proving that the operators $H^{\max}_{i,j}$ with the same $i$, but different $j$ are independent from each other. Let us look at the stabilizing operator $H^{\max}_{i,1}$. It is the only generator $H^{\max}_{i,j}$ that has a nontrivial matrix on the first qubit of $C_{i+2}$, so it has to be independent. Then, the generator $H^{\max}_{i,2}$ has its first nontrivial matrix at the second qubit of $C_{i+2}$ and the only generator that also has a nontrivial matrix on this qubit is $H^{\max}_{i,1}$. However $H^{\max}_{i,1}$ is independent, and so is $H^{\max}_{i,2}$. This can be done recursively for all $j$ which finally proves that $G_i^{\max}$ together with $H_{i,j}^{\max}$ form a set of $N$ independent generators.

\section*{acknowledgements} This work was supported by the (Polish) National Science Center through the SONATA BIS grant no. 2019/34/E/ST2/00369.

\end{document}